\documentclass[sigplan, cleveref, autoref, thm-restate, dvipsnames]{acmart}
\pdfoutput=1
\usepackage[capitalize]{cleveref}
\usepackage{bm}
\usepackage{stmaryrd}
\usepackage{mathtools}
\usepackage{macros}
\usepackage{xspace}
\usepackage{caption}
\usepackage{subcaption}
\usepackage{tikz}
\usetikzlibrary{petri,calc,backgrounds,positioning}
\usetikzlibrary{decorations,automata}
\usepackage{adjustbox}
\usepackage{xcolor}
\usepackage{pgfplots}
\pgfdeclarelayer{bg}    
\pgfsetlayers{bg,main}
\usetikzlibrary{decorations.pathreplacing,patterns,petri,calc}
\usetikzlibrary{backgrounds}
\usepackage{array}
\usepackage{balance}
\usepackage{thm-restate}

\pgfdeclarelayer{bg}    
\pgfsetlayers{bg,main}  

\setcopyright{none}



\acmConference{\color{white}}{}{}{}
\acmPrice{}
\acmDOI{}
\acmISBN{}
\acmYear{\color{white}}

\title{The complexity of soundness in workflow nets}

\author{Michael Blondin}
\affiliation{Universit\'{e} de Sherbrooke, Canada}
\email{michael.blondin@usherbrooke.ca}

\author{Filip Mazowiecki}
\affiliation{Max Planck Institute for Software Systems, Germany}
\email{filipm@mpi-sws.org}

\author{Philip Offtermatt}
\affiliation{Max Planck Institute for Software Systems, Germany}
\affiliation{Universit\'{e} de Sherbrooke, Canada}
\email{philip.offtermatt@usherbrooke.ca}


\begin{document}

\begin{abstract}
  Workflow nets are a popular variant of Petri nets that allow for algorithmic formal
  analysis of business processes. The central decision problems
  concerning workflow nets deal with soundness, where the initial and final configurations are specified.
  Intuitively, soundness
  states that from every reachable configuration one can reach
  the final configuration. We settle the widely open complexity of
  the three main variants of soundness: classical, structural and
  generalised soundness. The first two are EXPSPACE-complete, and,
  surprisingly, the latter is PSPACE-complete, thus computationally
  simpler.
\end{abstract}

\maketitle

\section{Introduction}
\label{sec:introduction}
\emph{Workflow nets} are a formalism that allows for the modeling of
business processes.
Specifically, they allow to formally represent workflow procedures in Workflow Management Systems (WFMSs)
(see \eg \cite[Section~4]{van1998application}, where Figure~6 shows a workflow net for the processing of complaints; and \cite[Section~3]{AL97} for details on modeling procedures).
Such a mathematical
representation enables the algorithmic formal analysis of their
behaviour. This is particularly relevant for large organisations that
seek to manage the workflow of complex business processes. Such
challenges have received, and continue to receive, intense academic
attention, \eg\ through the foundations track of the Business Process
Management Conference (BPM), and via a discipline coined
as \emph{process mining} and pioneered prolifically by Wil van der
Aalst\footnote{See \url{http://www.processmining.org}.}. In
particular, many tools, such as those integrated in the ProM
framework~\cite{DMVWA05}, can extract events from logs, \eg\ of
enterprise resource planning (ERP) systems, from which they synthesize
workflow nets (and other models) to be formally analyzed
(see \cite{vdAS11} for a book on the topic).

More formally, workflow nets form a subset of (standard) Petri
nets. They consist of \emph{places} that can contain resources (called
\emph{tokens}) which can be consumed and produced via \emph{transitions} in a
nondeterministic and concurrent fashion. Two designated places, namely
the \emph{initial place} $\initial$ and the \emph{final place}
$\output$, respectively model the initialisation and termination of a
business process. No token can be produced in the initial place, and
no token can be consumed from the final place.

A central property studied since the inception of workflow nets is
\emph{$1$-soundness}~\cite{AL97, van1998application}. Informally, quoting~\cite{AL97}, it states that
``{\it For any case, the procedure will terminate eventually [...]}''. More formally,
from the configuration with a single token in the initial place
$\initial$, every reachable configuration can reach the configuration
with a single token in the final place $\output$. For readers familiar
with computation temporal logic (CTL), $1$-soundness can be loosely
rephrased as $\initial \models \forall \mathsf{G}\, \exists
\mathsf{F}\, \output$. More generally, $k$-soundness states the same
but for $k$ tokens, \ie\ $(k \cdot \initial) \models \forall
\mathsf{G}\, \exists \mathsf{F}\, (k \cdot \output)$.

\paragraph{Classical soundness.}

Several variants of soundness have been considered in the literature
(see~\cite{AalstHHSVVW11} for a survey). The best-known
is \emph{classical soundness}. It states that a workflow net is
$1$-sound and that each transition is meaningful,
\ie\ each transition can be fired in at least one execution (often
called \emph{quasi-liveness}). It is well-known that deciding
classical soundness amounts to checking boundedness and liveness of a
slightly modified net. In particular, this means that classical
soundness is decidable since boundedness and liveness are decidable
problems. However, to the best of our knowledge, the (exact)
complexity of classical soundness remains widely open. It has been
suggested that classical soundness is EXPSPACE-hard. For example, the
author of~\cite{V99} mentions that ``\emph{IO-soundness is decidable
but also EXPSPACE-hard (\cite{V96})}'', yet~\cite{V96} merely states
the following:
\begin{quote}
  \emph{[I]t may be intractable to decide soundness. (For arbitrary
    [workflow]-nets liveness and boundedness are decidable but also
    EXPSPACE-hard [...])}.
\end{quote}
Furthermore, \cite[p.~38]{van1998application} claims that
EXPSPACE-hardness follows from the fact that ``\emph{deciding liveness
and boundedness is EXPSPACE-hard}'', which is attributed
to~\cite{CEP93}. However, \cite{CEP93} only mentions liveness to be
EXPSPACE-hard (which was known prior to~\cite{CEP93}).

The confusion arises from the fact that boundedness and liveness are
\emph{independently} EXPSPACE-hard problems, which suggests that
classical soundness must naturally be at least as hard. However, this
needs not be the case. For example, for a well-studied subclass of
Petri nets, called free-choice nets, testing \emph{simultaneously}
boundedness and liveness has lower complexity than testing both
properties independently\footnote{For free-choice nets: Boundedness is
EXPSPACE-complete since any Petri net can trivially be made
free-choice while preserving its reachability set up to projection;
liveness is coNP-complete~\cite[Thm.~4.28]{DE95}; and testing
liveness \emph{and} boundedness can be done in polynomial
time~\cite[Cor.~6.18]{DE95}.}~\cite{DE95}. Moreover, since liveness is
equivalent to the Petri net reachability problem~\cite{Hack76}, the
only (implicitly) known upper bound is not even primitive
recursive~\cite{LS19}. As a first contribution, we show that classical
soundness and $k$-soundness are in fact both EXPSPACE-hard and in
EXPSPACE, and hence EXPSPACE-complete. The upper bound is derived with
a fortiori surprisingly little effort by invoking known results on
coverability and so-called cyclicity. The hardness result is obtained
by a careful reduction from the reachability problem for reversible Petri nets~\cite{CLM76,MM82}.
There, we exploit subtle known results in
a technically challenging way.

\paragraph{Generalised and structural soundness.}

Among the variants of soundness catalogued by the survey of van der
Aalst et al.~\cite{AalstHHSVVW11}, \emph{generalised
soundness}~\cite[Def.~3]{van2003soundness} is the only fundamentally
distinct property (in particular,
see \cite[Fig.~7]{AalstHHSVVW11}). It asks whether a given workflow
net is $k$-sound \emph{for all} $k \geq 1$.
Generalised soundness, unlike classical soundness, preserves nice properties like composition~\cite{van2003soundness}.
The existential
counterpart of generalised soundness, where ``for all'' is replaced by ``for
some'', is known as \emph{structural
soundness}~\cite{barkaoui1998structural}.

It is a priori not clear whether generalised and structural soundness
are decidable, as the approach for deciding other types of soundness
reasons about $k$-soundness for a given or fixed number
$k$. Nonetheless, both problems have been shown
decidable~\cite{HSV04,TM05}. The two algorithms, and a subsequent
one~\cite{van2007verifying}, rely on Petri net reachability, which has
very recently been shown Ackermann-complete~\cite{LS19,Ler21,CO21}.

As for classical soundness, the computational complexity of
generalised and structural soundness remains open. In fact, we are not
aware of any complexity result. In this work, we prove that
generalised and structural soundness have much lower complexity than
Petri net reachability: they are respectively PSPACE-complete and
EXPSPACE-complete. In particular, the fact that generalised soundness
is simpler than classical soundness is arguably surprising: positive
instances of both problems require the given workflow net to be
bounded, but for generalised soundness, one can avoid explicitly
checking this EXPSPACE-complete property.

To derive the PSPACE membership, we introduce the notion
of \emph{strong soundness} which is (partly) defined in terms of a
relaxed reachability relation (sometimes known
as \emph{$\Z$-reachability} or \emph{pseudo-reachability}, \eg, see~\cite{Blo20}). Through
results on integer linear programming and bounded vectors reordering,
we prove that $k$-unsoundness of a workflow net must occur for a
``small'' number $k$. Furthermore, we show that it suffices to witness
such a $k$ for so-called $\Z$-bounded nonredundant nets, a more
restrictive property than (standard) boundedness. By building upon
these results, we establish the EXPSPACE membership of structural
soundness, and, in fact, effectively characterise the set of sound
numbers of workflow nets, which settles the open problem of~\cite{TM05}.

The hardness for PSPACE and EXPSPACE are respectively obtained via
reductions from the reachability problem for conservative Petri
nets~\cite{ConservativePN14}, and from $1$-soundness.

\paragraph{Contribution and organisation.}

In summary, we settle, after around two decades, the exact
computational complexity of the central decision problems for workflow
nets. This is achieved in the rest of this work, organised as
follows. In \Cref{sec:preliminaries}, we introduce general notation,
Petri nets, workflow nets and soundness. In \Cref{sec:classical}, we
prove that classical soundness is EXPSPACE-complete. In
\Cref{sec:stuff}, we provide bounds on vector reachability, which in
turn allows us to prove PSPACE-completeness of generalised soundness
(\Cref{sec:generalised}), and EXPSPACE-completeness of structural
soundness (\Cref{sec:structural}). In \Cref{sec:soundk}, we leverage
the previous results to give a characterisation of numbers $k$ for
which a workflow net is $k$-sound. Finally, we conclude in
\Cref{sec:conlusion}. Due to space constraints, some proofs are
deferred to an appendix.

\section{Preliminaries}
\label{sec:preliminaries}
We denote naturals and integers with the usual font: $n \in \N$ and $z \in \Z$. Given $i, j \in \Z$, we write $[i..j]$ for $\{i, i + 1, \ldots, j\}$. We use the bold font for vectors and matrices, \eg\ $\vec{a} = (a_1,\ldots,a_n) \in \Z^n$ and $\mat{A} \in \Z^{m \times n}$. Given $n \in \N$, we write $\vec{n}^d = (n, \ldots, n) \in \N^d$. We omit the dimension $d$ when it is clear from the context, \eg\ $\vec{0}$ denotes the null vector. We write $\vec{a}[i] = a_i$ and $\mat{A}[i, j]$ for matrix entries where $i \in [1..m]$ and $j \in [1..n]$. We write $\xx \le \vec{y}$ if $\xx[i] \le \vec{y}[i]$ holds for all $i \in [1..n]$. We write $\xx < \vec{y}$ if at least one inequality is strict. Given a vector $\vec{a} \in \Z^n$ or a matrix $\mat{A} \in \Z^{m \times n}$, we define the norms $\norm{\vec{a}} \defeq \max_{1 \le i \le n} \abs{\vec{a}[i]}$ and $\norm{\mat{A}} \defeq \max_{1 \le j \le m, 1 \le i \le n} \abs{\mat{A}[j,i]}$.

\subsection{Petri nets}

A \emph{Petri net} is a triple $\pn = (P, T, F)$ such that:
\begin{itemize}
\item $P$ and $T$ are disjoint finite sets whose elements are
  respectively called \emph{places} and \emph{transitions},

\item $F \colon ((P \times T) \cup (T \times P)) \to \N$ is the
  \emph{flow function}.
\end{itemize}

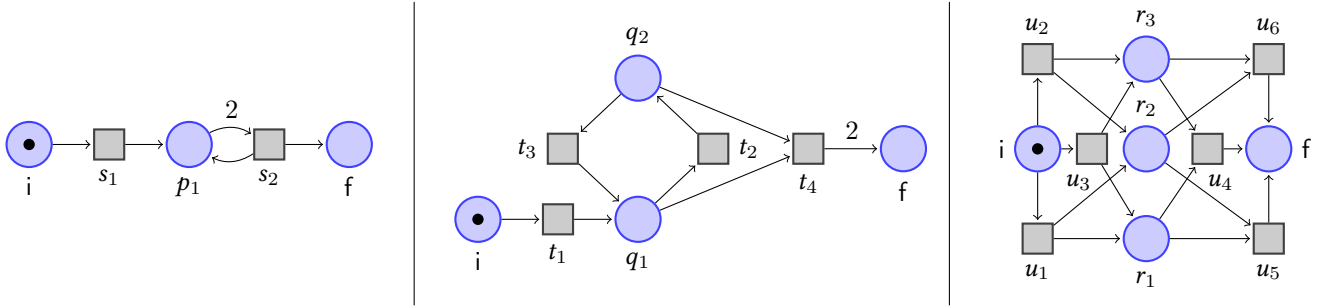
\begin{figure*}[h]
  \centering
  \begin{tabular}{m{5.25cm}|m{6.75cm}|m{4.75cm}}
    \begin{tikzpicture}
      \tikzstyle{place}=[circle,thick,draw=blue!75,fill=blue!20,minimum size=6mm]
      \tikzstyle{transition}=[rectangle,thick,draw=black!75,
        fill=black!20,minimum size=4mm]
      
      \node[place, tokens=1,           label=below:$\initial$] (i) {};
      \node[place, right = 1.5cm of i, label=below:$p_1$]      (p) {};
      \node[place, right = 1.5cm of p, label=below:$\output$]  (o) {};

      \node[transition, label=below:$s_1$] (e1) at ($(i)!0.5!(p)$) {}
      edge[pre]  (i)
      edge[post] (p)
      ;

      \node[transition, label=below:$s_2$] (e2) at ($(p)!0.5!(o)$) {}
      edge[pre,bend right] node[above] {$2$} (p)
      edge[post]            (o)
      edge[post, bend left] (p)
      ;
    \end{tikzpicture} &
    \quad%
    \begin{tikzpicture}
      \tikzstyle{place}=[circle,thick,draw=blue!75,fill=blue!20,minimum size=6mm]
      \tikzstyle{transition}=[rectangle,thick,draw=black!75,
        fill=black!20,minimum size=4mm]
      
      \node[place, tokens=1,             label=below:$\initial$] (i)  {};
      \node[place, right = 1.5cm of i,   label=below:$q_1$]      (p1) {};
      \node[place, above = 1.25cm of p1, label=above:$q_2$]      (p2) {};

      \node[transition, label=below:$t_1$] (e1) at ($(i)!0.5!(p1)$) {}
      edge[pre]  (i)
      edge[post] (p1)
      ;

      \node[transition, label=right:$t_2$] (e3) at ($(p1)!0.5!(p2) + (1,0)$) {}
      edge[pre]  (p1)
      edge[post] (p2)
      ;

      \node[transition, label=left:$t_3$] (e4) at ($(p1)!0.5!(p2) - (1,0)$) {}
      edge[pre]  (p2)
      edge[post] (p1)
      ;

      \node[place, right = 2cm of e3, label=below:$\output$] (o) {};
      \node[transition, label=below:$t_4$] (e2) at ($(e3)!0.5!(o)$) {}
      edge[pre] (p1)
      edge[pre] (p2)
      edge[post] node[above] {$2$} (o)
      ;
    \end{tikzpicture} &
    \quad%
    \begin{tikzpicture}
      \tikzstyle{place}=[circle,thick,draw=blue!75,fill=blue!20,minimum size=6mm]
      \tikzstyle{transition}=[rectangle,thick,draw=black!75,
        fill=black!20,minimum size=4mm]
      
      \node[place, label=left:$\initial$, tokens=1] (i) {};
      \node[place, above right=0.75cm and 1cm of i, label=above:$r_3$] (p1) {};
      \node[place, below right=0.75cm and 1cm of i, label=below:$r_1$] (p2) {};
      \node[place, label=above:$r_2$] (p3) at ($(p1)!0.5!(p2)$) {};
      \node[place, right=1cm of p3, label=right:$\output$]  (o) {};

      \node[transition, label=above:$u_2$] (e1) at (i |- p1) {}
      edge[pre]  (i)
      edge[post] (p1)
      edge[post] (p3)
      ;
      
      \node[transition, label=below:$u_1$] (e2) at (i |- p2) {}
      edge[pre]  (i)
      edge[post] (p2)
      edge[post] (p3)
      ;
      
      \node[transition, label=above:$u_6$] (e3) at (o |- p1) {}
      edge[post] (o)
      edge[pre]  (p1)
      edge[pre]  (p3)
      ;
      \node[transition, label=below:$u_5$] (e4) at (o |- p2) {}
      edge[post] (o)
      edge[pre]  (p2)
      edge[pre]  (p3)
      ;
      
      \node[transition, label={[xshift=5pt]below:$u_4$}] (e5)
      at ($(o)!0.5!(p3)$) {}
      edge[post] (o)
      edge[pre]  (p2)
      edge[pre]  (p1)
      ;
      
      \node[transition, label={[xshift=-5pt]below:$u_3$}]
      (e6) at ($(i)!0.5!(p3)$) {}
      edge[post] (p1)
      edge[post] (p2)
      edge[pre]  (i)
      ;
    \end{tikzpicture}
  \end{tabular}
  \caption{Three workflow nets, each marked with
    $\imarked{1}$.}\label{fig:examples}
\end{figure*}

A \emph{marking} is a vector $\m \colon P \to \N$ where $\m[p]$
indicates how many \emph{tokens} are contained in place $p$. We say
that a transition $t \in T$ is \emph{enabled} in $\m$ if $\m[p] \geq
F[p, t]$. Informally, $F[p, t]$ and $F[t, p]$ respectively correspond
to the amount of tokens to be consumed from and produced in place
$p$. Let $\pre{t}, \post{t} \in \N^p$ respectively denote the vectors
such that $\pre{t}[p] \defeq F[p, t]$ and $\post{t}[p] \defeq F[t,
p]$. Let $\effect{t} \defeq \post{t} - \pre{t}$ denote
the \emph{effect} of $t$. If transition $t$ is enabled in $\m$, then
$t$ may be \emph{fired}, which leads to the marking $\m'
\defeq \m + \effect{t}$. The latter is denoted by $\m
\trans{t} \m'$, or simply by $\m \trans{} \m'$ whenever we do not care
about the transition that led to $\m'$. We use a standard notation for markings, listing only nonzero values, \eg if
$P = \set{p_1, p_2}$, $\m[p_1] = 2$ and $\m[p_2] = 0$, then $\m = \set{p_1 \colon 2}$.

A \emph{run} is a sequence of transitions $\rho = t_1 \cdots t_n \in T^*$.
A run is enabled in $\m_0$ if there is a sequence of markings $\m_1,\ldots,\m_n$ such that
$\m_{i} \trans{t_i} \m_{i+1}$ for all $0 \le i < n$. If it is the case, then we denote this by
$\m_0 \trans{\rho} \m_n$, or $\m_0 \reach{} \m_n$ if $\rho$ is not important.
Given $\ell \in \N$, we say that $\rho$ is \emph{$\ell$-bounded} if
$\norm{\m_i} \le \ell$ for all $0 \le i \le n$. The support of a run is the set of transitions occurring in it, denoted
$\support{\rho} \defeq \set{t_1, \ldots, t_n}$.

We introduce a semantics where transitions can always be fired, and hence where markings may become negative. Formally, a \emph{$\Z$-marking} is a vector $\m \colon P \to \Z$. We write $\m \ztrans{t} \m'$ (or simply $\m \ztrans{} \m'$) if $\m' = \m + \effect{t}$. Given a run $\rho$, we define in the obvious way  $\ztrans{\rho}$ and ${\zreach}$. Note that markings are $\Z$-markings (with the domain restricted to $\N$). The definition of $\Z$-markings is mostly needed to use ${\zreach}$.

We define the \emph{absolute value} and \emph{norm} of a Petri net
$\pn = (P, T, F)$ by $\abs{\pn} \defeq |P| + |T|$ and $\norm{\pn}
\defeq \norm{F} + 1$, where $F$ is seen as a vector over $(P \times T)
\cup (T \times P)$. The \emph{size} of a Petri net is defined as
$\size{\pn} \defeq \abs{\pn} \cdot ( 1 + \log \norm{\pn})$. For some
complexity problems, we will be given a Petri net and some markings,
\eg $\m$ and $\m'$.  By the size of the input, we understand
$\size{\pn,\m,\m'} \defeq \size{\pn} + \log(\norm{\m} + 1) +
\log(\norm{\m'} + 1)$.

A transition $t$ is said to be \emph{quasi-live} from marking $\m$ if there
exists a marking $\m'$ such that $\m \reach \m'$ and $t$ is enabled in
$\m'$. A transition $t$ is said to be  \emph{live} from $\m$ if $t$ is
quasi-live from all $\m'$ such that $\m \reach \m'$. We say that a
Petri net $\pn$ is \emph{quasi-live} (resp.\ \emph{live}) from $\m$ if
each transition $t$ of $\pn$ is quasi-live (resp.\ live) from
$\m$. Informally, quasi-liveness states that no transition is useless,
and liveness states that transitions can always eventually be fired.

\begin{example}
  Consider the Petri net $\pn_\text{middle} = (P, T, F)$ illustrated
  in the middle of \Cref{fig:examples}. Places $P = \{\initial, q_1,
  q_2, \output\}$ and transitions $T = \{t_1, t_2, t_3, t_4\}$ are
  depicted respectively as circles and squares. The flow function $F$
  is depicted by arcs, where unit weights are omitted, and where arcs
  with weight zero are not drawn, \eg\ $F(\initial, t_1) = 1$,
  $F(t_1, \initial) = 0$, $F(t_4, \output) = 2$ and $F(\output, t_4) =
  0$. In particular, transitions $t_1$, $t_2$ and $t_3$ are quasi-live
  from marking $\imarked{1}$
  since \[\imarked{1} \trans{t_1} \{q_1 \colon
  1\} \trans{t_2} \{q_2 \colon 1\} \trans{t_3} \{q_1 \colon 1\}.\]
  However, as no other marking is reachable, transition $t_4$ is not
  quasi-live. Note that $t_2$ and $t_3$ are both live from
  $\imarked{1}$, while $t_1$ is not live since it can only be fired
  once.
\end{example}

\subsection{Workflow nets and soundness}

A \emph{workflow net} $\pn$ is a Petri net that satisfies the
following:
\begin{itemize}
\item there is a dedicated \emph{initial} place $\initial$ with
  $\pre{t}[\initial] = 0$ for every transition $t$ (cannot produce
  tokens in $\initial$);
          
\item there is a dedicated \emph{final} place $\output \neq \initial$
  with $\post{t}[\output] = 0$ for every transition $t$ (cannot
  consume tokens from $\output$);
          
\item each place and transition lies on at least one path from
  $\initial$ to $\output$ in the underlying graph of $\pn$, \ie\ the
  graph $(V, E)$ where $V \defeq P \cup T$ and $(u, v) \in E$ iff
  $F[u, v] > 0$.
\end{itemize}

Given $k \in \N$, we say that $\pn$ is $k$-sound iff $\imarked{k}
\reach \m$ implies $\m \reach \omarked{k}$, i.e.\ starting from $k$
tokens in the initial place, it is always possible to move the $k$
tokens into the final place. We say that $\pn$ is:
\begin{itemize}
\item \emph{classically sound} iff $\pn$ is $1$-sound and quasi-live from $\set{\initial : 1}$;

\item \emph{generalised sound} iff $\pn$ is
  $k$-sound for all $k > 0$;

\item \emph{structurally sound} iff $\pn$ is $k$-sound for some $k > 0$.
\end{itemize}

\begin{example}
  Consider the workflow nets $\pn_\text{left}$, $\pn_\text{middle}$
  and $\pn_\text{right}$ depicted respectively in
  \Cref{fig:examples}.

  Workflow nets $\pn_\text{left}$ and $\pn_\text{middle}$ are not
  $1$-sound since their only transition that can mark place $\output$
  is not quasi-live from $\imarked{1}$, namely $s_2$ and $t_4$. In
  particular, this means that both workflow nets are neither
  classically sound, nor generalized sound. Workflow net
  $\pn_\text{right}$ is $1$-sound, and in fact classically sound, as
  shown by the reachability graph of \Cref{fig:reach}.

  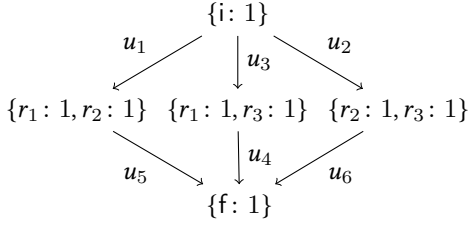
\begin{figure}[h]
  \begin{tikzpicture}[auto]
    \node (i) {$\imarked{1}$};

    \node[below left=20pt and 15pt of i]
    (r1r2) {$\{r_1 \colon 1, r_2 \colon 1\}$};

    \node[below=20pt of i]
    (r1r3) {$\{r_1 \colon 1, r_3 \colon 1\}$};

    \node[below right=20pt and 15pt of i]
    (r2r3) {$\{r_2 \colon 1, r_3 \colon 1\}$};
    
    \node[below right=20pt and 15pt of r1r2] (f) {$\fmarked{1}$};

    \path[->]
    (i) edge node[swap] {$u_1$} (r1r2)
    (i) edge node {$u_2$} (r2r3)
    (i) edge node {$u_3$} (r1r3)

    (r1r2) edge node[swap] {$u_5$} (f)
    (r2r3) edge node {$u_6$} (f)
    (r1r3) edge node {$u_4$} (f)
    ;
  \end{tikzpicture}
  \caption{Markings reachable from $\imarked{1}$ in $\pn_\text{right}$.}%
  \label{fig:reach}
\end{figure}
  
  In particular, this means that $\pn_\text{right}$ is structurally
  sound. Workflow net $\pn_\text{left}$ is not structurally sound as
  no matter the marking $\imarked{k}$ from which it starts, there is
  no way to empty place $p_1$ once it is marked. Workflow net
  $\pn_\text{middle}$ is $2$-sound, and hence structurally
  sound. Indeed, from $\imarked{2}$, the two tokens must enter $\{q_1,
  q_2\}$ from which they can escape via $\{q_1 \colon 1, q_2 \colon
  1\}$ by firing $t_4$, reaching marking $\fmarked{2}$.

  Workflow net $\pn_\text{right}$ is not $2$-sound, and hence not
  generalised sound. Indeed, we have $\imarked{2} \trans{u_1 u_2
  u_4} \{r_2 \colon 2, \output \colon 1\}$ and no transition is
  enabled in the latter marking.
\end{example}

\section{Classical soundness}
\label{sec:classical}
As mentioned in the introduction, classical soundness is decidable,
but its complexity has not yet been established. Let us recall why
decidability holds. We say that a Petri net $\pn$ is \emph{bounded}
from marking $\m$ if there exists $b \in \N$ such that $\m \reach \m'$
implies $\m' \leq \vec{b}$. Otherwise, $\pn$ is \emph{unbounded} from
$\m$. It is well-known that unboundedness holds iff there exist
markings $\m' < \m''$ such that $\m \reach \m' \reach \m''$. The
\emph{short-circuit net} $\pn_{sc}$ of a workflow net $\pn$ is $\pn$
extended with a transition $t_{sc}$ such that $F[\output, t_{sc}] =
F[t_{sc}, \initial] = 1$ (and $0$ for other entries relating to
$t_{sc}$). Informally, the short-circuit net allows to restore the
system upon completion, \ie\ by moving a token from $\output$ to
$\initial$.

\begin{figure}
  \begin{tabular}{m{3.65cm}|m{4cm}}
    \hspace*{-0.75cm} 
    \begin{tikzpicture}[transform shape, scale=0.8]
      \tikzstyle{place}=[circle,thick,draw=blue!75,fill=blue!20,minimum size=6mm]
      \tikzstyle{transition}=[rectangle,thick,draw=black!75,
        fill=black!20,minimum size=4mm]
      
      \node[place, label=left:$\initial$, tokens=1] (i) {};
      \node[place, above right = 1cm and 1cm of i, label=above:$r_3$] (p1) {};
      \node[place, below right = 1cm and 1cm of i, label=below:$r_1$] (p2) {};
      \node[place, label=above:$r_2$] (p3) at ($(p1)!0.5!(p2)$) {};
      \node[place, right=1cm of p3, label=right:$\output$]  (o) {};

      \node[transition, label=above:$u_2$] (e1) at (i |- p1) {}
      edge[pre]  (i)
      edge[post] (p1)
      edge[post] (p3)
      ;
      
      \node[transition, label=below:$u_1$] (e2) at (i |- p2) {}
      edge[pre]  (i)
      edge[post] (p2)
      edge[post] (p3)
      ;
      
      \node[transition, label=above:$u_6$] (e3) at (o |- p1) {}
      edge[post] (o)
      edge[pre]  (p1)
      edge[pre]  (p3)
      ;
      \node[transition, label=below:$u_5$] (e4) at (o |- p2) {}
      edge[post] (o)
      edge[pre]  (p2)
      edge[pre]  (p3)
      ;
      
      \node[transition, label={[xshift=5pt]below:$u_4$}] (e5)
      at ($(o)!0.5!(p3)$) {}
      edge[post] (o)
      edge[pre]  (p2)
      edge[pre]  (p1)
      ;
      
      \node[transition, label={[xshift=-5pt]below:$u_3$}]
      (e6) at ($(i)!0.5!(p3)$) {}
      edge[post] (p1)
      edge[post] (p2)
      edge[pre]  (i)
      ;

      \node[transition, below=15pt of p2, label=below:$t_{sc}$] (tsc) {};

      \path[->]
      (o)   edge[out=-45, in=0, looseness=1.5]    node {} (tsc)
      (tsc) edge[out=180, in=-135, looseness=1.5] node {} (i)
      ;
    \end{tikzpicture}%
    \hspace*{-5pt} 
    &
    \hspace*{-5pt} 
    \begin{tikzpicture}[auto, transform shape, scale=0.8]
      \node[font=\small] (i) {$\imarked{1}$};

      \node[below left=20pt and 10pt of i, font=\small]
      (r1r2) {$\{r_1 \colon 1, r_2 \colon 1\}$};

      \node[below=20pt of i, font=\small]
      (r1r3) {$\{r_1 \colon 1, r_3 \colon 1\}$};

      \node[below right=20pt and 10pt of i, font=\small]
      (r2r3) {$\{r_2 \colon 1, r_3 \colon 1\}$};
      
      \node[below right=20pt and 10pt of r1r2, font=\small] (f) {$\fmarked{1}$};

      \path[->]
      (i) edge node[swap] {$u_1$} (r1r2)
      (i) edge node {$u_2$} (r2r3)
      (i) edge node {$u_3$} (r1r3)

      (r1r2) edge node[swap] {$u_5$} (f)
      (r2r3) edge node {$u_6$} (f)
      (r1r3) edge node {$u_4$} (f)

      (f) edge[out=-15, in=15, looseness=2.85]
      node[xshift=-35pt, yshift=50pt] {$t_{sc}$} (i)
      ;
    \end{tikzpicture}
  \end{tabular}
  \caption{\emph{Left}: Short-circuit net of the rightmost workflow
    net from \Cref{fig:examples}. \emph{Right}: Its markings reachable
    from $\imarked{1}$.}%
  \label{fig:sc}
\end{figure}
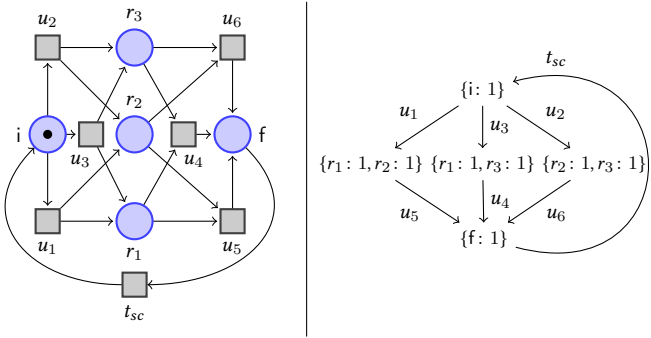

For example, the left side of \Cref{fig:sc} illustrates a
short-circuit net $\pn_{sc}$. By inspecting the graph of markings
reachable from $\imarked{1}$ in $\pn_{sc}$, we see that $\pn_{sc}$ is
live and bounded, \ie\ it is always possible to (re)fire any
transition, and each place is bounded by $b \defeq 1$ token. It turns
out that liveness and boundedness characterize classical soundness:

\begin{proposition}[{\cite[Lemma~8]{AL97}}]\label{prop:charac:sound}
  A workflow net $\pn$ is classically sound iff $\pn_{sc}$ is live and
  bounded from $\imarked{1}$.
\end{proposition}

Decidability of classical soundness follows
from \Cref{prop:charac:sound}. Indeed, boundedness can be tested in
EXPSPACE~\cite{Rac78}, and liveness is decidable since it reduces to
reachability~\cite[Thm~5.1]{Hack76} which is
decidable~\cite{May81}. However, the liveness problem is hard for the
reachability problem~\cite[Thm~5.2]{Hack76}, which was recently shown
Ackermann-complete~\cite{LS19,Ler21,CO21}. In this section, we first
give a slightly different characterization not involving
liveness which yields EXPSPACE membership. Then, we show that classical
soundness is EXPSPACE-hard, and hence EXPSPACE-complete, via a
reduction from the reachability problem for so-called reversible Petri
nets.

\subsection{EXPSPACE membership}

Let us reformulate the characterization of \Cref{prop:charac:sound} so
that it deals with another property than liveness, namely
``cyclicity''. We say that a Petri net is \emph{cyclic} from a marking
$\m$ if $\m \reach \m'$ implies $\m' \reach \m$, i.e.\ it is always
possible to go back to $\m$. For example, the short-circuit net
$\pn_{sc}$, illustrated on the left of \Cref{fig:sc}, is cyclic since
each marking reachable from $\imarked{1}$ can reach $\omarked{1}$,
which in turn can reach $\imarked{1}$.

Rather than directly considering classical soundness, we first
consider $1$-soundness. The characterization
of \Cref{prop:charac:sound} can be adapted to this problem as follows:

\begin{restatable}{proposition}{propCharacOneSound}\label{prop:charac:1sound}
  A workflow net $\pn$ is $1$-sound iff $\pn_{sc}$ is bounded and
  transition $t_{sc}$ is live from $\imarked{1}$.
\end{restatable}

From the previous proposition, we prove the following.

\begin{lemma}\label{lem:1sound-charact}
  A workflow net $\pn = (P, T, F)$ is $1$-sound iff $\pn_{sc}$ is
  bounded and cyclic from $\imarked{1}$, and some transition $t \in T$
  satisfies $\pre{t} = \imarked{1}$.
\end{lemma}

\begin{proof}
  $\Rightarrow$) Let $\pn$ be $1$-sound. Since $\imarked{1} \reach
  \omarked{1}$ and $\initial \neq \output$,
   some $t \in T$ satisfies $\pre{t} = \imarked{1}$. By \Cref{prop:charac:1sound}, from $\imarked{1}$,
  $\pn_{sc}$ is bounded and $t_{sc}$ is live. It remains to show that
  $\pn_{sc}$ is cyclic. Let $\imarked{1} \reach \m$. By liveness of
  $t_{sc}$, there is a marking $\m'$ such that $\m \reach \m'$ and
  $\m'$ enables $t_{sc}$. Note that $\pre{t_{sc}} =
  \omarked{1}$. If $\m' > \omarked{1}$, that is, $\m' = \omarked{1} +
  \n$ with $\n > \vec{0}$, then we obtain $\imarked{1} \reach
  \omarked{1} + \n \trans{t_{sc}} \imarked{1} + \n$, and hence
  boundedness is violated. Thus, by boundedness and liveness of
  $t_{sc}$, $\m \reach \m' = \omarked{1} \trans{t_{sc}} \imarked{1}$,
  which proves cyclicity.

  $\Leftarrow$) Assume $\pn_{sc}$ is bounded and cyclic from
  $\imarked{1}$, and that some $t \in T$ is as described. By
  \Cref{prop:charac:1sound}, it suffices to show that $t_{sc}$ is live
  from $\imarked{1}$. Let $\m \in \N^P$ be such that $\imarked{1}
  \reach \m$ in $\pn_{sc}$. We either have $\m = \imarked{1}$ or
  $\m[\initial] = 0$, as otherwise $\pn_{sc}$ is unbounded. If $\m =
  \imarked{1}$, we can fire $t$ and obtain a marking where $\initial$
  is empty. Thus, assume w.l.o.g.\ that $\m[\initial] = 0$. By
  cyclicity, we have $\m \trans{\pi} \imarked{1}$ for some
  $\pi$. Since $t_{sc}$ is the only transition that produces tokens in
  place $\initial$, transition $t_{sc}$ must appear in $\pi$. Hence,
  $t_{sc}$ is live.
\end{proof}

Since classical soundness amounts to quasi-liveness and $1$-soundness,
we obtain the following corollary.

\begin{corollary}\label{cor:charac:sound:cyc}
  A workflow net $\pn$ is classically sound iff $\pn_{sc}$ is
  quasi-live, bounded and cyclic from $\imarked{1}$.
\end{corollary}

\begin{theorem}\label{thm:class-easy}
  Both $1$-soundness and classical soundness are in EXPSPACE.
\end{theorem}

\begin{proof}
  Checking whether a transition $t$ satisfies $\pre{t} = \imarked{1}$
  can be carried in polynomial time. The
  other properties of \Cref{lem:1sound-charact} for $1$-soundness,
  namely boundedness and cyclicity, belong to
  EXPSPACE~\cite{Rac78,bouziane1997cyclic}.

  For quasi-liveness, we proceed as follows. The \emph{coverability
    problem} asks whether given a Petri net and two markings $\m,
  \m'$, there exists a marking $\m'' \geq \m'$ such that $\m \reach
  \m''$. This problem belongs to EXPSPACE~\cite{Rac78}. Recall that
  quasi-liveness asks whether for each transition $t \in T \cup
  \set{t_{sc}}$, it is the case that $\imarked{1} \reach \m$ for some
  some marking $\m$ that enables $t$, i.e.\ such that $\m \geq
  \pre{t}$. The latter is a coverability question. Hence,
  quasi-liveness amounts to $\abs{T} + 1$ coverability queries, which
  can be checked in EXPSPACE.
\end{proof}

We further show that the previous result can be extended to
$k$-soundness through the following lemma.

\begin{lemma}\label{lem:scale-net}
  Given a workflow net $\pn$ and $k > 0$, one can compute, in
  polynomial time, a workflow net $\pn'$ with $\norm{\pn'}
  = \norm{\pn} + \log(k)$ such that, for all $c > 0$, $\pn$ is
  $ck$-sound iff $\pn'$ is $c$-sound.
\end{lemma}

\begin{proof}
  Let $\pn = (P, T, F)$. We define $\pn' \defeq (P', T', F')$ that rescales
  everything by $k$. Formally, we add two new places that are the new
  initial and final places $P' \defeq P \cup
  \set{\initial',\output'}$. We denote by $\initial$ and $\output$ the
  previous initial and final places. We add two new transitions
  $t_\initial$ and $t_\output$ defined by:
  \begin{alignat*}{5}
    \pre{t_\initial}[\initial'] &= 1\
    &\text{and }&& \pre{t_\initial}[p]\ &= 0\
    &&\text{for } p \neq \initial', \\
    \post{t_\initial}[\initial] &= k\
    &\text{and }&& \post{t_\initial}[p]\ &= 0\
    &&\text{for } p \neq \initial, \\
    \pre{t_\output}[\output] &= k\
    &\text{and }&& \pre{t_\output}[p]\ &= 0\
    &&\text{for } p \neq \output, \\
    \post{t_\output}[\output'] &= 1\
    &\text{and }&& \post{t_\output}[p]\ &= 0\
    &&\text{for } p \neq \output'.
  \end{alignat*}
  It is straightforward that $\pn'$ satisfies the lemma.
\end{proof}

\begin{corollary}\label{cor:k-sound}
  The $k$-soundness problem is in EXPSPACE.
\end{corollary}

\begin{proof}
  It suffices to invoke \Cref{lem:scale-net} with $c = 1$, and test
  $1$-soundness of the resulting workflow net via
  \Cref{thm:class-easy}.
\end{proof}

\subsection{EXPSPACE-hardness}

Let us now establish EXPSPACE-hardness of classical soundness. We will
need the forthcoming lemma that essentially states that so-called
reversible Petri nets can count up to (or down from) a doubly
exponential number. Formally, we say that a Petri net $\pn = (P, T,
F)$ is \emph{reversible} if each transition of $\pn$ has an
inverse, \ie\ for every $t \in T$, there exists $t^{-1} \in T$ such
that $\pre{(t^{-1})} = \post{t}$ and $\post{(t^{-1})} = \pre{t}$.
Note that for reversible Petri nets, it is the case that $\m \reach \m'$ if and only if $\m' \reach \m$.
To emphasise this, we will sometimes write $\m \rreach \m'$.

\begin{lemma}[{\cite[Lemma~3]{MM82}}]\label{lem:rev:en}
  Let $\pn$ be a reversible Petri net and let $\m$ and $\m'$ be two
  markings. Let $n \defeq \size{\pn,\m,\m'}$. There exists $c_n \in
  2^{2^{\bigO(n)}}$ such that if $\m \reach{} \m'$ then $\m
  \trans{\rho} \m'$ for a $c_n$-bounded run $\rho$.
\end{lemma}

\begin{lemma}[{\cite[reformulation of Lemma~6 and Lemma~8]{MM82}}]\label{lem:count:gadget}
  Let $n \in \N$ and $c_n \in 2^{2^{\bigO(n)}}$. There exists a reversible Petri net $\pn_n = (P_{n}, T_{n}, F_{n})$
  with four distinguished places $s,c,f, b \in P_{n}$. Given the two markings $\m_n \defeq \{s \colon 1, c \colon 1\}$ and $\m_n'
  \defeq \{f \colon 1, c \colon 1, b \colon c_n\}$, the following holds for all $\m$:
  \begin{enumerate}
    \item $\m_n \rreach \m_n'$;

    \item $\m_n \rreach \m$ and $\m[f] > 0$ implies $\m = \m_n'$;

    \item $\m \rreach \m_n'$ and $\m[s] > 0$ implies $\m = \m_n$;
      
    \item if $\m < \m_n'$ and $\m[f] = 0$ then no transition can be fired from $\m$;

    \item for all $p \in P_n$ there exists $\m_n \rreach \m$ s.t.\ $\m[p] > 0$.
  \end{enumerate}
  Furthermore, $\pn_n$ is: of polynomial size in $n$; constructible in polynomial time in $n$; and quasi-live both from $\m_n$ and $\m_n'$.
\end{lemma}

\begin{restatable}{theorem}{thmClassHard}\label{thm:class-hard}
  The classical soundness and $1$-soundness problems are EXPSPACE-hard.
\end{restatable}

\begin{proof}
    We give a reduction from the reachability problem for reversible
    Petri nets. This problem is known to be
    EXPSPACE-complete~\cite{CLM76,MM82}. Let $\pn = (P, T, F)$ be a reversible
    Petri net, and let $\m, \m'$ be two markings for which we would like
    to know whether $\m \reach \m'$ in $\pn$.

    Let $n \defeq \size{\pn, \m, \m'}$. Let $c_n$ be
    the value given by \Cref{lem:rev:en} for $n$.
    Let
    $\pn_n= (P_n, T_n, F_n)$ be the Petri net given
    by \Cref{lem:count:gadget} for $c_n$.

    \begin{figure*}[!h]
    \centering
    \begin{tikzpicture}[node distance=1.25cm, transform shape, scale=0.9]
        \tikzstyle{place}=[circle,thick,draw=blue!75,fill=blue!20,minimum size=6mm]
        \tikzstyle{transition}=[rectangle,thick,draw=black!75,
        fill=black!20,minimum size=4mm]
        \tikzstyle{red place}=[place,draw=red!75,fill=red!20]
        \tikzstyle{copy place}=[place,draw=\phaseTwoColor!75,fill=\phaseTwoColor!20]

        \node [red place,tokens=1, label=left:{$\initial$}] (i)                                    {};
        \node [red place, label=above:{$p_{\mathrm{start}}$}] (q0) [right = 2cm of i]                      {};
        \node [red place, label=above:{$p_{\mathrm{inProgress}}$},right= 2cm of q0] (c1) {};
        \node [red place, label=above:{$p_{\mathrm{cover}}$}] (q1) [right= 4cm of c1]                      {};
        \node [red place, label=right:{$\output$}] (q2) [right = 3cm of q1]                      {};

        \node [place, label=above:{$p_1$},above=1.5cm of q0] (p1)        {};
        \node [copy place, label=right:{$\overline{p}_1$},below=1.5cm of q0] (notp1)  {};
        \node [place, label=above:{$p_2$},right=2cm of p1] (p2)        {};
        \node [copy place, label=below:{$\overline{p}_2$},right=2cm of notp1] (notp2)  {};
        \node [red place, label=left:{$p_{\mathrm{simple}}$},above=3.5cm of i] (p3)        {};

        \node [transition] (t1) [label=left:{$t_{\mathrm{hard}}T_2^* t_{\mathrm{start}}$}] at (notp1 -| i) {}
        edge[pre,\phaseTwoDecColor] (i)
        edge[post,\phaseTwoDecColor] (q0)
        edge[post,\phaseTwoDecColor] node[above] {$c_n$} (notp1)
        edge[post,bend right=40,\phaseTwoDecColor] node[above] {$c_n$} (notp2)
        ;
\begin{pgfonlayer}{bg}
        \node [transition] (te1) [label=left:{$t_{\mathrm{simple}}$}] at (p1 -| i) {}
        edge[pre,gray,densely dotted,thick] (i)
        edge[post,gray,densely dotted,thick] node[above] {$\norm{\pn}$} (p1)
        edge[post,bend left=40,gray,densely dotted,thick] node[above] {$\norm{\pn}$} (p2)
        edge[post,gray,densely dotted,thick] node[right] {$\norm{\pn}$} (notp1)
        edge[post,gray,densely dotted,thick] node[above] {} (p3)
        ;

        \path[gray, densely dotted, thick, ->]
        (te1) edge[out=-15, in=140, looseness=2] node[xshift=38.5pt, yshift=-42.5pt] {$\norm{\pn}$} (notp2)
        ;

\end{pgfonlayer}
        \node [transition, label=above:{$t_{\m}$}] (tm) at ($(q0)!0.5!(c1)$) {}
        edge[post] (c1)
        edge[pre] (q0)
        edge[pre] (notp1)
        edge[post] (p1)
        ;
        \node [transition, label=above:{$t_{\m'}$}] (tm') at ($(q1)!0.5!(c1) + (0.5,0.1)$) {}
        edge[post,bend left] (q1)
        edge[pre,bend right] (c1)
        edge[post,bend left] (notp2)
        edge[pre] (p2)
        ;
        \node [transition, label=below:{$\ \ t_{\m'}^{-1}$}] (tm'm1) at ($(q1)!0.5!(c1) - (0.5,0.1)$) {}
        edge[pre,bend right] (q1)
        edge[post,bend left] (c1)
        edge[pre] (notp2)
        edge[post] (p2)
        ;
        \node [transition, label=above:{$t_{\mathrm{reach}}T_3^*t_{\mathrm{end}}$}] (tc) at ($(q1)!0.5!(q2)$) {}
        edge[pre,\phaseTwoDecColor] (q1)
        edge[post,\phaseTwoDecColor] (q2)
        edge[pre,bend left,\phaseTwoDecColor] node[above] {$c_n$} (notp2)
        edge[pre,bend left=50,\phaseTwoDecColor] node[above] {$c_n$} (notp1)
        ;
\begin{pgfonlayer}{bg}
        \node [transition] (te2) [label = right:{$t_{\mathrm{simple}2}$}] at (p3 -| q2) {}
        edge[pre,bend right=5,gray,densely dotted,thick] (p3)
        edge[pre,bend right=10,gray,densely dotted,thick] node[above] {$\norm{\pn}$} (p1)
        edge[pre,gray,densely dotted,thick] node[above] {$\norm{\pn}$} (p2)
        edge[pre,bend left=5,gray,densely dotted,thick] node[left,near start] {$\norm{\pn}$} (notp1)
        edge[pre,bend left=30,gray,densely dotted,thick] node[left, near start] {$\norm{\pn}$} (notp2)
        edge[post,gray,densely dotted,thick] (q2)
        ;
\end{pgfonlayer}
%
%
    \end{tikzpicture}
    \caption{A workflow net $\pn'$ which is classically sound iff $\m \reach \m'$ in the reversible Petri net $\pn = (P,T)$. In the example, $P = \set{p_1, p_2}$, $\m = (1,0)$ and $\m' = (0,1)$. The original places are blue, their copies are green, and other new places are red. We omit the transitions in $T_1$ that originated from $T$ (recall that these transitions are modified to consume and produce tokens also in green places), and we omit the place $p_{\mathrm{canFire}}$ (used only to allow transitions in $T_1$ to fire). We only sketch transitions in $T_2$ and $T_3$ (and some other transitions), by writing the intuitive meaning of the gadgets that add/remove $c_n$ tokens (these ``transitions'' are marked with a different color).
    The transition $t_{\mathrm{hard}}$ initiates the bottom part of $\pn'$ (by filling the green places with $c_n$ tokens) that checks $\m \reach \m'$. The transition $t_{\mathrm{simple}}$ initiates the top part of $\pn'$. We denote transitions in the top part with dotted gray color. This part is rather trivial and its only purpose is to ensure quasi-liveness of transitions in $T_1$ (by filling blue and green places with $\lVert \pn \rVert$ tokens).
}
    \label{fig:soundness-expspace-hard}
\end{figure*}

    We construct a workflow net
    $\pn' = (P', T', F')$ such that $\pn'$ is classically sound if and only if $\m \trans{*} \m'$
    in $\pn$. To avoid any confusion, we will denote markings in $\pn'$ by $\n$, $\n'$, etc. 
  
    The construction will ensure that
    \begin{align}\label{eq:main}
    \m \trans{*} \m' \text{ in  } \pn \text{ iff } \pn' \text{ is classically sound}.
    \end{align}
    Moreover, $1$-soundness of $\pn'$ will imply $\m \trans{*} \m'$, which will prove that both classical soundness and $1$-soundness are EXPSPACE-hard.

    Formally,
    the set of places $P'$ consists of:
    $P$; its disjoint copy $\overline{P} \defeq \set{\overline{p} \mid p \in
      P}$; seven extra places
\[
      \{\initial, \output, p_{\mathrm{start}}, p_{\mathrm{inProgress}}, p_{\mathrm{cover}}, p_{\mathrm{simple}}, p_{\mathrm{canFire}}\};
\]
two disjoint copies of $P_n$ (from \Cref{lem:count:gadget}), with one copy of $b$ removed.
One of the copies will be marked with $\heartsuit$ to avoid any confusion, thus we write \eg\ $p^\heartsuit \in P_n^\heartsuit$.
The two places $b$ and $b^\heartsuit$ are merged into a single place denoted $b$.

Before presenting the transitions, we would like to emphasise that, intuitively,
    place $\overline{p} \in \overline{P}$ will contain a
    ``budget'' of tokens that is an upper bound on how many more tokens can be present in $p$.
    Most of the time, for every marking $\m$ and place $p \in P$, we will keep $\m[p] + \m[\overline{p}] = c_n$ as an invariant.

    In \Cref{fig:soundness-expspace-hard}, we present the most relevant parts of $\pn'$.
    Formally, the set of transitions is divided into four subsets $T' = T_1 \cup T_2 \cup T_3 \cup T_4$.
    Transitions will be defined by giving $\pre{t'}[p]$ and $\post{t'}[p]$. The values are zero on unmentioned places.
    
    First, for every transition $t \in T$, we define $t' \in T_1$ by:
    \begin{itemize}
     \item $\pre{t'}[p] \defeq \pre{t}[p]$ and $\post{t'}[p] \defeq \post{t}[p]$ for all $p \in P$;
     \item $\pre{t'} [\overline{p}] \defeq \post{t}[p]$ and $\post{t'} [\overline{p}] \defeq \pre{t}[p]$ for all $p \in P$;
     \item $\pre{t'}[p_{\mathrm{canFire}}] = \post{t'}[p_{\mathrm{canFire}}] \defeq 1$.
    \end{itemize}
It is easy to see that since $\pn$ is a reversible Petri net, for every transition in $T_1$, its reverse is also in $T_1$. We will say that $T_1$ is reversible. Notice that, for all $t' \in T_1$ and $p \in P$, the sum of tokens in $p$ and $\overline{p}$ does not change under $t'$.

Second, for every $t \in T_n$, we add $t' \in T_2$ such that:
    \begin{itemize}
     \item $\pre{t'}[p] \defeq \pre{t}[p]$ and $\post{t'}[p] \defeq \post{t}[p]$ for all $p \in P_n$;
     \item $\pre{t'}[\overline{p}] \defeq \pre{t}[b]$ and $\post{t'}[\overline{p}] \defeq \post{t}[b]$ for all $\overline{p} \in \overline{P}$.
    \end{itemize}
    Intuitively, places in $\overline{P}$ behave as $b$ to initialise the budget of $c_n$ tokens.
    Similarly, for every $t^\heartsuit \in T_{n}^\heartsuit$, we add $t' \in T_3$ such that:
    \begin{itemize}
     \item $\pre{t'}[p^\heartsuit] \defeq \pre{t}[p^\heartsuit]$ and $\post{t'}[p^\heartsuit] \defeq \post{t}[p^\heartsuit]$ for all $p^\heartsuit \in P^\heartsuit_n$;
     \item $\pre{t'}[\overline{p}] \defeq \pre{t}[b]$ and $\post{t'}[\overline{p}] \defeq \post{t}[b]$ for all $\overline{p} \in \overline{P}$.
    \end{itemize}
    Note that since $\pn_n$ is reversible, both $T_2$ and $T_3$ are reversible.
    
    The set $T_4$ consists of the ten remaining transitions
    \[
    \set{t_{\mathrm{hard}}, t_{\mathrm{start}}, t_{\m}, t_{\m'}, t_{\m'}^{-1}, t_{\mathrm{isEmpty}}, t_{\mathrm{reach}}, t_{\mathrm{reach}}^{-1}, t_{\mathrm{simple}}, t_{\mathrm{simple}2}}.
    \]
    Intuitively, the first two are needed to initialise places in $\overline{P}$ with $c_n$ tokens; the next three transitions respectively
    add $\m$, $\m'$ and $-\m'$ to $P$; the next three transitions transfer tokens towards the final places; and the last two transitions are needed for quasi-liveness. Formally,
    \begin{itemize}
    \item $\pre{t_{\mathrm{hard}}}[\initial] = \post{t_{\mathrm{hard}}}[s] = \post{t_{\mathrm{hard}}}[c] \defeq 1$;

    \item $\pre{t_{\mathrm{start}}}[f] = \pre{t_{\mathrm{start}}}[c] = \post{t_{\mathrm{start}}}[p_{\mathrm{start}}] \defeq 1$;

    \item $\post{t_{\m}}[p] = \pre{t_{\m}}[\overline{p}] \defeq \m[p]$ for all $p \in P$; and $\pre{t_{\m}}[p_{\mathrm{start}}] = \post{t_{\m}}[p_{\mathrm{inProgress}}] = \post{t_{\m}}[p_{\mathrm{canFire}}] \defeq 1$;

    \item $\pre{t_{\m'}}[p] = \post{t_{\m'}}[\overline{p}] \defeq \m'[p]$ for all $p \in P$; $\pre{t_{\m'}}[p_{\mathrm{inProgress}}] = \pre{t_{\m}}[p_{\mathrm{canFire}}] = \post{t_{\m'}}[p_{\mathrm{cover}}] \defeq 1$; and $t_{\m'}^{-1}$ is its reverse transition;

    \item $\pre{t_{\mathrm{reach}}}[p_{\mathrm{cover}}] = \post{t_{\mathrm{reach}}}[f^{\heartsuit}] = \post{t_{\mathrm{reach}}}[c^\heartsuit] \defeq 1$; and $t_{\mathrm{reach}}^{-1}$ is its reverse transition;

    \item $\pre{t_{\mathrm{end}}}[s^{\heartsuit}] = \pre{t_{\mathrm{end}}}[c^{\heartsuit}] = \post{t_{\mathrm{end}}}[\output] \defeq 1$;
    \item $\pre{t_{\mathrm{simple}}}[\initial]=\post{t_{\mathrm{simple}}}[p_{\mathrm{simple}}] = \post{t_{\mathrm{simple}}}[p_{\mathrm{canFire}}] \defeq 1$; and $\post{t_{\mathrm{simple}}}[p] = \post{t_{\mathrm{simple}}}[\overline{p}] \defeq \norm{\pn}$ for all $p \in P$;
    \item $\pre{t_{\mathrm{simple}2}}[p] = \pre{t_{\mathrm{simple}2}}[\overline{p}] \defeq \norm{\pn}$ for all $p \in P$; and $\pre{t_{\mathrm{simple}2}}[p_{\mathrm{simple}}] = \pre{t_{\mathrm{simple}2}}[p_{\mathrm{canFire}}] = \post{t_{\mathrm{simple}2}}[\output] \defeq 1$.
    \end{itemize}

    Recall that $P \subseteq P'$ and that $\m$ is a marking on $P$. To ease the notation, we will assume that $\m$ is a marking on $P'$ (with $0$ tokens in places from $P' \setminus P$).
    
    We are ready to prove \Cref{eq:main}. Notice that for every reachable configuration $\set{\initial \colon 1} \trans{\rho} \n$ the value $\n[p_{\mathrm{canFire}}]$ is always equal to $\n[p_{\mathrm{simple}}]$ or $\n[p_{\mathrm{inProgress}}]$ (depending on whether the first transitions of $\rho$ is $t_{\mathrm{simple}}$ or $t_{\mathrm{hard}}$). For readability, we omit the value of $p_{\mathrm{canFire}}$ in the markings of $\pn'$.
    
    $\Leftarrow$) Suppose that $\pn'$ is $1$-sound (we will not rely on $\pn'$ being quasi-live). By \Cref{lem:count:gadget}~(1), we know that
    \[
    \set{\initial \colon 1} \trans{t_{\mathrm{hard}}} \set{s \colon 1, c \colon 1} \reach \set{f \colon 1, c \colon 1, b \colon c_n} + \sum_{\overline{p} \in \overline{P}} \set{\overline{p} \colon c_n}.
    \]
    Let us denote the last marking by $\n$. Notice that
    \[
    \n \trans{t_{\mathrm{start}}t_{\m}} \set{p_{\mathrm{inProgress}} \colon 1, b \colon c_n} + \m + \sum_{\overline{p} \in \overline{P}} \set{\overline{p} \colon c_n - \m[p]}.
    \]
    We denote the latter marking by $\n'$. Since $\pn'$ is $1$-sound, $\n' \trans{\rho} \set{\output \colon 1}$ for some run $\rho$. This is possible if  $t_{\mathrm{reach}}$ was fired at least once in $\rho$. Let $\n_1 \trans{t_{\mathrm{reach}}} \n_2$ be the last time $t_{\mathrm{reach}}$ was fired in $\rho$. We claim that $\n_2 = \set{f^\heartsuit \colon 1, c^\heartsuit \colon 1, b \colon c_n} + \sum_{\overline{p}} \set{\overline{p} \colon c_n}$. Indeed, it has to be that 
    \[
    \n_2 \trans{\rho'} \set{s^\heartsuit \colon 1, c^\heartsuit \colon 1} \trans{t_{\mathrm{end}}} \set{\output \colon 1},
    \]
    where $\rho'$ uses transition only from $T_3$. By \Cref{lem:count:gadget}~(4), this is possible only if $\n_2$ is as claimed. 
    Let $\rho''$ be the prefix of the run $\rho$ from $\n'$ such that it ends in $\n_1$. Finally, $\rho''$, when restricted to $P$, witnesses reachability for $\m \reach \m'$.

 $\Rightarrow$) Suppose that $\m \reach \m'$. The proof of $1$-soundness is very technical and can be found in the appendix. In a nutshell, recall that $T_1$, $T_2$ and $T_3$ are reversible, and for $t_{\m'}, t_{\mathrm{reach}} \in T_4$ we include their reverse transitions. This allows us to revert any configuration to a configuration from which it is easy to define a run to $\set{\output \colon 1}$.
 
To conclude this implication, we need to prove that $\pn'$ is quasi-live. Indeed, from the proof of $1$-soundness it is easy to see that
$\m \reach \m'$ implies that all transitions are fireable, with the possible exception of transitions from $T_1$. However,
\[
\set{\initial \colon 1} \trans{t_{\mathrm{simple}}} \set{p_{\mathrm{simple}} \colon 1} + \sum_{p \in P} \set{p \colon \norm{\pn}, \overline{p} \colon \norm{\pn}}.
\]
From the latter configuration, any transition of $T_1$ is fireable.

    Finally, observe that $\pn'$ is a workflow net. Indeed, by taking $t_{\mathrm{simple}}$ we put tokens in $P$, and by taking 
    $t_{\mathrm{simple}2}$ we can put tokens in $\overline{P}$. Each place from copies in $\pn_n$ is on a path from $\initial$ to $\output$ by  \Cref{lem:count:gadget} (5). The remaining places are clearly on such a path by definition (see \Cref{fig:soundness-expspace-hard}).
\end{proof}

\section{Bounds on vector reachability}
\label{sec:stuff}
In this section, we present technical results that will be helpful to establish complexity bounds in the forthcoming sections. It is well-known that Petri nets are complex due to their nonnegativity constraints. Namely, markings are over $\N$ (not $\Z$), which blocks transitions from being fired whenever the amount of tokens would drop below zero. By lifting this restriction, \ie\ allowing markings over $\Z$, transitions cannot be blocked and we obtain a provably simpler model (\eg see~\cite{HaaseH14}). We recall known results that provide bounds on reachability problems for vectors over $\Z$. Based on these results, we will derive useful bounds for the next sections.

\subsection{Integer linear programs}

Given positive natural numbers $n, m > 0$, let $\mat{A} \in \Z^{m \times n}$ be an integer matrix, $\vec{b} \in \Z^{m}$ an integer vector and $\xx = (x_1,\ldots,x_n)^\transpose$ a vector of variables.
We say that $G \defeq \mat{A} \cdot \xx \geq \vec{b}$ is an \emph{$(m \times n)$-ILP}, that is, an integer linear program (ILP) with $m$ inequalities and $n$ variables.
The set of solutions of $G$ is \[\solutions{G} \defeq \{\vec{\mu} \in \Z^n \mid \mat{A} \cdot \vec{\mu} \geq \vec{b}\},\]
and the set of natural solutions is $\nsolutions{G} \defeq \solutions{G} \cap \N^n$.
We will only be interested in the natural solutions $\nsolutions{G}$ but sometimes we will need to refer to $\solutions{G}$. We shall assume that these sets are equal, by implicitly adding a new inequality for each variable specifying that it is greater or equal to $0$.

Often it is convenient to write an equality constraint, \eg $x - y = 0$. This can be simulated by two inequalities, so we will allow to define $G$ both with equalities and inequalities.

We introduce some notation about \emph{semi-linear} sets from~\cite{CH16} to obtain bounds on the sizes of solutions to ILPs.
A set of vectors is called \emph{linear} if it is of the form $L(\vec{b},P) = \set{\vec{b} + \lambda_1 \vec{p}_1 + \ldots + \lambda_k \vec{p}_k \mid \lambda_1, \ldots, \lambda_k \in \N}$, where
$\vec{b} \in \Z^n$ is a vector and $P = \set{\vec{p}_1,\ldots, \vec{p}_k} \subseteq \Z^n$ is a finite set of vectors.
A set is called \emph{hybrid linear} if it is of the form
$L(B,P) = \bigcup_{\vec{b} \in B} L(\vec{b},P)$ for a finite set of vectors $B = \set{\vec{b}_1\ldots,\vec{b}_\ell} \subseteq \Z^n$.

The \emph{size} of a finite set of vectors $B$ and of an $(m \times
n)$-ILP $G$ are defined respectively as $\norm{B} \defeq \max_{\vec{b} \in
B}\norm{\vec{b}}$ and $\norm{G} \defeq \norm{\mat{A}} + \norm{\vec{b}}
+ m + n$.

\begin{lemma}[\cite{GS78}, presentation adapted from~{\cite[Prop.~3]{CH16}}]\label{lem:ilp-decomp}
    Let $G$ be an $(m \times n)$-ILP.
    It is the case that $\solutions{G} = \bigcup_{i \in I} L(B_i, P_i)$, where
    $max_{i \in I} \norm{B_i} \leq \norm{G}^{\bigO(n \log n)}$.
\end{lemma}

For the forthcoming lemmas, recall that $\vec{c} = (c, \ldots, c)$.

\begin{lemma}\label{lem:small-ilp}
    Let $G$ be an $(m \times n)$-ILP. There exists a number
    $c \le \norm{G}^{\bigO(n \log n)}$ such
    that for all $\vec{\mu} \in \nsolutions{G}$, there is some $\vec{\mu}' \in
        \nsolutions{G}$ such that $\vec{\mu}' \leq \vec{\mu}$ and $\vec{\mu}' \leq \vec{c}$.
\end{lemma}

\begin{proof}
  Recall that we can assume $\solutions{G} = \nsolutions{G}$. By
  \Cref{lem:ilp-decomp}, $\solutions{G} = \bigcup_{i \in I} L(B_i,
  P_i)$. We set $c \defeq \max_{i \in I} \norm{B_i}$. Let $\vec{\mu}
  \in \nsolutions{G}$. There exist $i \in I$ and $\vec{b} \in B_i$
  such that $\vec{\mu} \in L(\vec{b}, P_i)$. Note that $\vec{p} \geq
  \vec{0}$ for all $\vec{p} \in P_i$. Hence, we have $\vec{b} \in
  \nsolutions{G}$, $\vec{b} \leq \vec{\mu}$ and $\vec{b} \leq
  \vec{c}$. Thus, we can set $\vec{\mu}' \defeq \vec{b}$.
\end{proof}

\begin{restatable}{lemma}{lemSmallSolution}\label{lem:small-solution}
  Let $G = \mat{A} \cdot \vec{x} \geq \vec{b}$ be an $(m \times
  n)$-ILP, where $\vec{b} \ge \vec{0}$. There exists $c \le
  \norm{G}^{\bigO((m+n)\log(m + n))}$ such that the following holds.
  For every $\vec{\mu} \in \nsolutions{G}$, there exists $\vec{\nu}
  \in \nsolutions{G}$ such that $\vec{\nu} \le \vec{\mu}$, $\vec{\nu}
  \le \vec{c}$, and $\mat{A} \cdot \vec{\nu} \le \mat{A} \cdot
  \vec{\mu}$.
\end{restatable}

\subsection{Steinitz Lemma}

Let us recall the Steinitz Lemma~\cite{steinitz1913bedingt} based on the presentation of~\cite{eisenbrand2019proximity}.

    \begin{figure*}[!h]
\centering
 \begin{tikzpicture}[every node/.append style={circle, inner sep=1.3pt},extended line/.style={shorten >=-#1,shorten <=-#1},
 extended line/.default=1cm, transform shape, scale=0.9]
\begin{axis}[grid=both, xmax=15,ymax=7.5,
                  xmin=-2.5,ymin=-2.5,
                  axis lines=middle,
                        xtick distance={1},
                        ytick distance={1},
              ticks=none,
              axis equal image
              ]
\node[fill] (v) at (axis cs:0,0) {};
\node[fill] (v0) at (axis cs:-1,2) {};
\node[fill] (v1) at (axis cs:1,3) {};
\node[fill] (v2) at (axis cs:3,4) {};
\node[fill] (v3) at (axis cs:5,5) {};
\node[fill] (v4) at (axis cs:7,6) {};
\node[fill] (v5) at (axis cs:9,5) {};
\node[fill] (v6) at (axis cs:11,4) {};
\node[fill] (v7) at (axis cs:12,3) {};
\node[fill,label=below:{$\vec{z}$}] (v8) at (axis cs:13,2) {};

\path[->]
    (v) edge[->,line width=1pt,>=stealth,red] node[black,left] {$\vec{x}_0$} (v0)
    (v0) edge[->,line width=1pt,>=stealth,red] node[black,above] {$\vec{x}_1$} (v1)
    (v1) edge[->,line width=1pt,>=stealth,red] node[black,above] {$\vec{x}_2$} (v2)
    (v2) edge[->,line width=1pt,>=stealth,red] node[black,above] {$\vec{x}_3$} (v3)
    (v3) edge[->,line width=1pt,>=stealth,red] node[black,above] {$\vec{x}_4$} (v4)
    (v4) edge[->,line width=1pt,>=stealth,red] node[black,above] {$\vec{x}_5$} (v5)
    (v5) edge[->,line width=1pt,>=stealth,red] node[black,above] {$\vec{x}_6$} (v6)
    (v6) edge[->,line width=1pt,>=stealth,red] node[black,above right] {$\vec{x}_7$} (v7)
    (v7) edge[->,line width=1pt,>=stealth,red] node[black,above right] {$\vec{x}_8$} (v8)
;

\begin{pgfonlayer}{bg}    
    \draw[blue!20,line width=2cm,extended line=0.5cm] (axis cs:0,0) -- (axis cs:13,2);
    \end{pgfonlayer}
\end{axis}
\end{tikzpicture}
\hspace*{1.5cm}
\begin{tikzpicture}[every node/.append style={circle, inner sep=1.3pt},extended line/.style={shorten >=-#1,shorten <=-#1},
 extended line/.default=1cm, transform shape, scale=0.9]

\begin{axis}[grid=both, xmax=15,ymax=7.5,
                  xmin=-2.5,ymin=-2.5,
                  axis lines=middle,
                        xtick distance={1},
                        ytick distance={1},
              ticks=none,
              axis equal image
              ]
\node[fill] (v) at (axis cs:0,0) {};
\node[fill] (v0) at (axis cs:-1,2) {};
\node[fill] (v1) at (axis cs:1,1) {};
\node[fill] (v2) at (axis cs:2,0) {};
\node[fill] (v3) at (axis cs:3,-1) {};
\node[fill] (v4) at (axis cs:5,0) {};
\node[fill] (v5) at (axis cs:7,1) {};
\node[fill] (v6) at (axis cs:9,2) {};
\node[fill] (v7) at (axis cs:11,1) {};
\node[fill,label=below:{$\vec{z}$}] (v8) at (axis cs:13,2) {};

\path[->]
    (v) edge[->,line width=1pt,>=stealth,red] node[black,left] {$\vec{x}_0$} (v0)
    (v0) edge[->,line width=1pt,>=stealth,red] node[black,above] {$\vec{x}_5$} (v1)
    (v1) edge[->,line width=1pt,>=stealth,red] node[black,above right] {$\vec{x}_8$} (v2)
    (v2) edge[->,line width=1pt,>=stealth,red] node[black,below left] {$\vec{x}_7$} (v3)
    (v3) edge[->,line width=1pt,>=stealth,red] node[black,below right] {$\vec{x}_2$} (v4)
    (v4) edge[->,line width=1pt,>=stealth,red] node[black,above] {$\vec{x}_3$} (v5)
    (v5) edge[->,line width=1pt,>=stealth,red] node[black,above] {$\vec{x}_4$} (v6)
    (v6) edge[->,line width=1pt,>=stealth,red] node[black,above] {$\vec{x}_6$} (v7)
    (v7) edge[->,line width=1pt,>=stealth,red] node[black,above] {$\vec{x}_1$} (v8)
;

\begin{pgfonlayer}{bg}    
    \draw[blue!20,line width=2cm,extended line=0.5cm] (axis cs:0,0) -- (axis cs:13,2);
    \end{pgfonlayer}
\end{axis}
\end{tikzpicture}
\caption{An example of \Cref{lem:extended-steinitz} in dimension $d = 2$. The vectors $\vec{x}_0,\ldots,\vec{x}_n$ form a path from $\vec{0}$ to $\vec{z}$. The colored background highlights points that are within some bounded distance from the line $\vec{0}$ to $\vec{z}$ (the bound 
 depends on $d$ and $\vec{x}_i$, but not on $\vec{z}$). In the right picture, the vectors are reordered so that they all fit within the bound. The additional constraints are that: the first vector $\vec{x}_0$ remains first ($\pi(0) = 0$); and, intuitively, that the points are getting closer to $\vec{z}$ ($0 \le c_0 \le c_1 \le \ldots \le c_n$).}\label{fig:steinitz}
    \end{figure*}
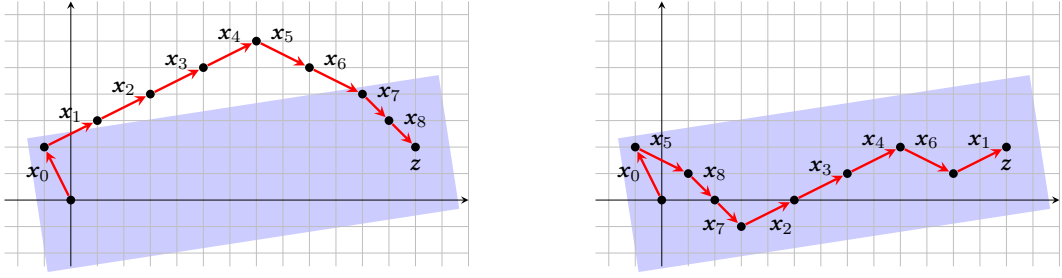

\begin{lemma}\label{lemma:steinitz}
    Let $\vec{x}_1, \dots, \vec{x}_n \in \R^d$ be such that $\sum_{i=1}^{n} \vec{x}_i = \vec{0}$ and $\norm{\vec{x}_i} \leq 1$ for all $i$.
    There exists a permutation $\pi$ on $[1..n]$ such that
    \begin{align*}
      \norm{\sum_{j=1}^i \vec{x}_{\pi(j)}} \leq d
      && \text{for all } i \in [1..n].
    \end{align*}
\end{lemma}

The following formulation of the lemma, which is depicted graphically in
\Cref{fig:steinitz}, will be more convenient for us.

\begin{restatable}{lemma}{lemExtendedSteinitz}\label{lem:extended-steinitz}
    Let $\vec{x}_0,\vec{x}_1, \dots, \vec{x}_n \in \Z^d$,
    $b \defeq \max_{j=0}^n \norm{\vec{x}_j}$, and $\vec{z} \defeq \sum_{j=0}^n \vec{x}_j$.
    There exists a permutation $\pi$ of $[0..n]$ such that: $\pi(0) = 0$; and there exist $0 \le c_0  \le c_1 \le \ldots \le c_n$, where
    \begin{align*}
      \norm{\sum_{j=0}^i \vec{x}_{\pi(j)} - c_i \cdot \vec{z}} \le b(d+2)
      && \text{for all } i \in [0..n].
    \end{align*}
\end{restatable}

\section{Generalised soundness}
\label{sec:generalised}
A Petri net $\pn$ is \emph{$\Z$-bounded} from a marking $\m$ if there
exists $b \in \N$ such that $\m \zreach \m' \geq \vec{0}$ implies
$\m' \le \vec{b}$ (\ie we replace ${\reach}$ with ${\zreach}$ in the
definition of boudedness). Otherwise, we say that $\pn$
is \emph{$\Z$-unbounded}. Observe that being $\Z$-bounded does not
mean that the set of reachable markings is bounded by below, but only
from above.

Let $k \ge 0$. We say that $\pn$ is \emph{strongly $k$-sound} if for every $\m \in \N^P$ such that
$\imarked{k} \zreach \m$, it holds that $\m \reach \omarked{k}$. Note
that every strongly $k$-sound net is also $k$-sound.

The aim of the next three subsections is to prove the following theorem.

\begin{theorem}\label{theorem:pspaceupper}
  Generalised soundness is in PSPACE.
\end{theorem}

The proof has two parts. First, we prove that if there is a $k$ for which the net is not $k$-sound, then there is also such a $k$ bounded exponentially. Second, we prove that $k$-soundness for exponentially bounded $k$ can be verified in PSPACE.

\subsection{Nonredundant workflow nets}

Fix a workflow net $\pn = (P, T, F)$. We say that a place $p \in P$ is \emph{nonredundant} if there exists $k \in \N$ such that $\imarked{k} \reach \m$ and $\m[p] > 0$. By removing a redundant place $p$ from $\pn$, we mean removing $p$ from $P$ and all transitions $t \in T$ such that $(\pre{t})[p] > 0$. With the remaining transitions restricted to the domain $P \setminus \set{p}$, we obtain a new workflow net $\pn' \defeq (P\setminus \set{p}, T')$. It is clear that $\pn$ is $k$-sound if and only if $\pn'$ is $k$-sound for all $k\in \N$. Thus, in particular, this procedure preserves generalised soundness.

It will be convenient to assume that all places in the studied workflow nets are nonredundant. At first, it might seem that this requires coverability checks for every place. However, since the number of initial tokens is arbitrary, finding redundant places amounts to a simple polynomial-time saturation procedure. 
More details can be found in~\cite[Thm.~8, Def.~10, Sect.~3.2]{HSV04} (and in the appendix). We will call workflow nets without redundant places \emph{nonredundant workflow nets}\footnote{The results in~\cite{HSV04} deal with batch workflow nets, which are in particular nonredundant workflow nets.}. To summarise we conclude the following.

\begin{restatable}{proposition}{propBatch}\label{claim:batch}
  Given a workflow net $\pn$, one can identify and remove all
  redundant places from it in polynomial time. The resulting workflow
  net $\pn'$ is nonredundant. Moreover, $\pn$ is $k$-sound if and only
  if $\pn'$ is $k$-sound for all $k\in \N$.
\end{restatable}

In the following lemma, intuitively, we show that the initial budget is small for nonredundant workflow nets.

\begin{lemma}\label{lem:placecover}
    Let $\pn = (P, T, F)$ be a nonredundant workflow net and let $p \in P$ be a place. There exists $k < (\norm{T} + 2)^{\abs{T}}$ such that
    $\imarked{k} \reach \m$ and $\m[p] > 0$.
\end{lemma}

\begin{proof}
    A transition $t$ \emph{increases} a place $p'$ if $\effect{t}[p'] > 0$. We say that a run $\rho$ \emph{increases} $p'$ if there exists $t \in \support{\rho}$ that increases $p'$.
    For the proof of the lemma, we assume that $p \neq \initial$, as otherwise it suffices to define $k = 1$.

    We prove that for all run $\imarked{k'} \trans{\rho} \m'$, 
    there is a run $\pi$
    such that: $\support{\pi} = \support{\rho}$, and $\imarked{k} \trans{\pi} \m$ for some $k < (\norm{T} + 2)^{n}$ and $\m$, 
    where $\m[p'] \geq 1$ for all places $p'$ increased by $\rho$.
    Note that, since $\pn$ is a nonredundant workflow net,
    if we exhibit such a run then we are done as there exists $\rho$ that increases $p$.
    
    Let $\imarked{k'} \trans{\rho} \m'$.
    The proof is by induction on $n$, where $\support{\rho} = \{t_1, \ldots, t_n\}$.
    Assume $n = 1$. The only transition used by $\rho$ is $t_1$, which increases $p$.
    Recall that $\norm{T}$ is the maximal number occurring on any arc of $\pn$.
    Since workflow nets start with tokens only in place $\initial$,
    we must have $\imarked{\norm{T}} \geq \pre{\pi}$.
    It suffices to define $\pi \defeq t_1$ and $k \defeq \norm{T} < (\norm{T} + 2)$.

    For the induction step, assume $n > 1$ and that the lemma holds for $n-1$. Let $\rho_{n-1}$ be the longest prefix of $\rho$ such that $\support{\rho_{n-1}} = \set{t_1,\ldots, t_{n-1}}$.
    The induction hypothesis for $\rho_{n-1}$ yields $k_{n-1} < (\norm{T} + 2)^{n-1}$,
    and $\pi_{n-1}$ with $\support{\pi_{n-1}} = \set{t_1,\dots,t_{n-1}}$. Let $\imarked{k_{n-1}} \trans{\pi_{n-1}} \m_{n-1}$.
    Note that $\support{\pre{t_{n}}} \allowbreak \subseteq \support{\post{\pi_{n-1}}} \cup \set{i}$ since $\rho$ is a run, where
    $t_{n}$ is fired.
    By repeating $\norm{T}+1$ times the run $\pi_{n-1}$,
    we get
    \[
    \imarked{(k_{n-1} + 1) \cdot (\norm{T} + 1)} \reach
    \imarked{\norm{T}+1} + (\norm{T}+1) \cdot \m_{n-1}.
    \]
    To ease the notation, let $\n \defeq \imarked{\norm{T}+1} + (\norm{T}+1) \cdot \m_{n}$.
    By definition of $\m_{n-1}$, it holds that $\n[p'] \geq \norm{T} + 1$ for all $p' \in \post{\pi}$.
    Furthermore, we can fire $t_{n}$ from $\n$. Let $\n \trans{t_{n}} \m$.
    To conclude, consider a place $p'$ increased by $\rho$. If it is increased by one of the transitions $t_1,\ldots,t_{n-1}$, then after firing $t_n$ at least one token was left in $p'$. Otherwise, $p'$ is increased by $t_n$. In both cases, we have $\m[p] \ge 1$. It remains to observe that $k = (k_{n-1} +1) \cdot (\norm{T} + 1) < (\norm{T}+2)^n$.
\end{proof}

\subsection{Unsoundness occurs for small numbers}\label{subsec:smallk}
Recall a result by van Hee et al.\ that establishes a connection between reachability relations ${\zreach}$ and ${\reach}$.

\begin{lemma}[adaptation of~{\cite[Lemma~12]{HSV04}}]\label{lem:zreachnreach}
    Let $\pn$ be a nonredundant workflow net, and let $\m$ be a marking for which there exists $k \geq 0$ satisfying
    $\imarked{k} \zreach \m$. There exists $\ell \geq 0$
    such that $\imarked{k+\ell} \reach \m + \omarked{\ell}$.
\end{lemma}

Note that \cref{lem:zreachnreach} is an easy consequence of the definition of nonredundancy. Namely, it suffices to put ``enough budget'' in each place so that the run under ${\zreach}$ becomes a run under ${\reach}$.
We restate the result to give a bound on $\ell$.

\begin{lemma}\label{lem:small-ell}
Let $\pn = (P, T, F)$ be a nonredundant workflow net.
    Let $k$ and $\vec{m} \in \N^P$ be such that $\imarked{k} \zreach \m$.
    There exist $\ell \leq (\norm{T} + 2)^{\abs{T}} \cdot \max(\norm{T}, k)\cdot \abs{P}(\abs{P}+2)$
    and $\m' \in \N^P$ such that $\imarked{\ell} \reach \m'$
    and $\imarked{\ell+k} \reach \m + \m'$.
\end{lemma}

\begin{proof}
    Let $\rho = t_1 t_2 \cdots t_n$ be such that $\imarked{k} \ztrans{\rho} \m$. Let us define $\vec{x}_0 \defeq \imarked{k}$ and $\vec{x}_j \defeq \effect{t_j}$ for all $j \in [1..n]$.
    By \cref{lem:extended-steinitz}, we can assume that the transitions $t_j$ are ordered so that there exist $c_0, \ldots, c_n \ge 0$ where
\begin{gather*}
    \norm{\imarked{k} + \sum_{j=1}^i \effect{t_j} - c_i \m} \le \max(\norm{T},k) \cdot (\abs{P} + 2),
\end{gather*}
for all $i \in [0..n]$.
Since $\m \ge 0$, we get for all $p \in P$:
\begin{gather}\label{eq:enough}
\left(\imarked{k} + \sum_{j=1}^i \effect{t_j} \right)\![p] \ge - \max(\norm{T},k) \cdot (\abs{P} + 2).
\end{gather}

    By \cref{lem:placecover}, there exists $\ell \leq (\norm{T} + 2)^{\abs{T}}$
    such that for every place $p$ there is a run $\imarked{\ell} \trans{\pi_p} \m_{p}$ with $\m_{p}[p] > 0$.
    Thus, to put $\max(\norm{T},k) \cdot (\abs{P}+2)$ tokens in all places, it suffices to repeat $\max(\norm{T},k) \cdot (\abs{P}+2)$ times the run $\pi_p$ for every $p \in P$.
    This requires $\ell \le (\norm{T} + 2)^{\abs{T}} \cdot \max(\norm{T},k) \cdot \abs{P}(\abs{P}+2)$ tokens. Let $\m'$ be the marking obtained afterwards.
    By~\eqref{eq:enough}, $\vec{m}'$ allows to fire $\rho$. Therefore, we obtain
    $\imarked{\ell} \reach \m'$
    and $\imarked{\ell+k} \reach \m + \m'$ as required.
\end{proof}

This lemma allows us to focus on ${\zreach}$ instead
of ${\reach}$.

\begin{lemma}\label{lem:unssound-unsound}
Let $\pn = (P, T, F)$ be a nonredundant workflow net.
It is the case that $\pn$ is generalised sound iff it is strongly $k$-sound for all $k \ge 0$.
    Moreover, if $\pn$ is not strongly $k$-sound, then there exists $k' \leq k + (\norm{T} + 2)^{\abs{T}} \cdot \max(\norm{T}, k)\cdot \abs{P}(\abs{P} +2)$ such that $\pn$ is not $k'$-sound.
\end{lemma}

\begin{proof}
    The ``if'' implication is trivial. Indeed, if $\pn$ is not $k$-sound then it cannot be strongly $k$-sound.
    
    To prove the ``only if'' implication, assume that $\pn$ is not strongly $k$-sound. We show that
    there exists $k'$ such that $\pn$ is not $k'$-sound. We will also prove the promised bound on $k'$.
    Since $\pn$ is not strongly $k$-sound, there must be some $\m \in \N^P$ and $\pi$ such that
    $\imarked{k} \ztrans{\pi} \m$ and $\m \not \reach \omarked{k}$.
    By \cref{lem:small-ell}, there exists $\ell \leq (\norm{T} + 2)^{\abs{T}} \cdot \max(\norm{T}, k)\cdot \abs{P} (\abs{P} + 2)$
    and $\m'$ such that
    $\imarked{\ell} \reach \m'$ and $\imarked{\ell + k} \reach \m + \m'$.
    If $\pn$ is not $\ell$-sound, then we are done.
    Otherwise, if $\pn$ is $\ell$-sound, then
    it must hold that $\m' \reach \fmarked{\ell}$.
    So, $\imarked{\ell + k} \reach \m + \m' \reach \m + \fmarked{\ell}$.
    Recall that $\m \not \reach \fmarked{k}$.
    Thus, $\m + \fmarked{\ell} \not \reach \fmarked{\ell + k}$. We are done since this means that $\pn$ is not $(\ell + k)$-sound.
\end{proof}

In the remainder of this section, we will show that if there exists some $k$ such that $\pn$ is not strongly $k$-sound,
then $k$ is at most exponential in $\abs{\pn}$.
We define an ILP which is closely related to the markings reachable from at least one initial number of tokens in $\pn$.
Essentially, the ILP will encode that there exists $k > 0$ and $\m \geq \0$ such that $\imarked{k} \zreach \m$. This can be done since only ``firing counts'' matter, \ie\ $\m \ztrans{\pi} \m'$ implies $\m \ztrans{\pi'} \m'$ for any permutation $\pi'$ of $\pi$.

Let $\pn = (P, T, F)$ be a workflow net. We define $\text{ILP$_{\pn}$} \defeq \mat{A} \cdot \vec{x} \geq \vec{0}$ as an ILP with $\abs{P} + \abs{T} + 1$ inequalities and $\abs{T} + 1$ variables. The variables of ILP $_{\pn}$ are $\vec{x} \defeq (\kappa,\tau_1,\dots,\tau_{\abs{T}})$. For ease of notation, we write $\ttau = (\tau_1,\dots,\tau_{\abs{T}})$. We assume an implicit bijection between $T$ and $[1..\abs{T}]$, \ie\ for every $t \in T$ there is a unique $i$ such that: $\ttau[t] = \tau_i$.
The matrix $\mat{A}$ is defined by the following inequalities:
\begin{enumerate}
\item\label{eq:A1} $\kappa + \sum_{t \in T} \ttau[t] \cdot \effect{t}[\initial] \geq 0$,

\item\label{eq:A2} $\kappa \ge 1$,

\item\label{eq:A3} $\sum_{t \in T} \ttau[t] \cdot \effect{t}[p] \geq 0$ for all $p \in P \setminus \set{\initial}$,

\item\label{eq:A4} $\tau_i \ge 0$ for all $i \in [1..\abs{T}]$.
\end{enumerate}

The first two inequalities concern the initial ``budget'' $k$ of tokens in $\initial$ which is represented by $\kappa$. Intuitively, $\kappa \ge 1$ has to be at least as much as $\ttau$ consumes from the initial place.
The last two inequalities guarantee that we obtain a marking over $\N^P$ and that the ``firing count'' is over $\N^T$.

Let $\vec{\mu} \colon \vec{x} \to \N$ be a solution to $\text{ILP}_{\pn}$.
We define
\[
\marking{\vec{\mu}} \defeq \imarked{\mu(\kappa)} + \sum_{t \in T}^{\abs{T}} \mu(\tau_j) \cdot \effect{t_j}.
\]
The following claim follows by definition of ILP$_{\pn}$ and ${\zreach}$.
\begin{claim}\label{lem:characterization}
    Let $\m \in \N^P$ and $k > 0$.
    It holds that $\imarked{k} \zreach \m$
    iff
    there exists a solution $\vec{\mu}$ to ILP$_{\pn}$
    such that $\marking{\vec{\mu}} = \m$ and $\vec{\mu}[\kappa] = k$.
\end{claim}

We conclude this part with the following bound.

\begin{lemma}\label{thm:first-unsound-small}
Let $\pn$ be a nonredundant workflow net.
    If $\pn$ is strongly $i$-sound for all $1 \leq i < k$,
    and not strongly $k$-sound, then $k \leq c$, where $c$ is the bound from \Cref{lem:small-solution} for ILP$_{\pn}$.
\end{lemma}

\begin{proof}
    For the sake of contradiction, assume that $k \gr c$ is as in the statement.
    Since $\pn$ is not strongly $k$-sound, there exists a marking $\m \in \N^P$ such that
    $\imarked{k} \zreach \m$ and $\m \not \reach \fmarked{k}$.
    By \cref{lem:characterization},
    there exists a solution $\vec{\mu}$ to ILP$_\pn$ such that
    $\marking{\vec{\mu}} = \m$ and $\vec{\mu}[\kappa] = k$.
    By \cref{lem:small-solution}, there exists a solution $\vec{\mu}' \leq \vec{\mu}$ to ILP$_{\pn}$ such that
    $\vec{\mu}'[\kappa] \le c < k = \vec{\mu}[\kappa]$ and $\mat{A} \vec{\mu}' \leq \mat{A} \vec{\mu}$, where $\mat{A}$ is the underlying matrix of ILP$_{\pn}$.
    The latter inequality implies $\marking{\vec{\mu}'} \leq \marking{\vec{\mu}}$.

    Consider the vector $\vec{\pi} \defeq \vec{\mu} - \vec{\mu}'$.
    We prove that $\vec{\pi}$ is a solution to ILP$_\pn$.  Since $\vec{\mu}' \leq \vec{\mu}$ we know that $\vec{\pi}$ is nonnegative.
    The inequalities of $\mat{A}$ are satisfied since $\mat{A} \vec{\pi} \geq \vec{0} \equiv \mat{A} \vec{\mu} \geq \mat{A} \vec{\mu}'$ and $\vec{\mu}'[\kappa] \le c < \vec{\mu}[\kappa]$.
    Thus, $\vec{\pi}$ is a solution to ILP$_\pn$.

    By \cref{lem:characterization}, $\imarked{\vec{\mu}'[\kappa]} \zreach  \marking{\vec{\mu}'}$ and $\imarked{\vec{\pi}[\kappa]} \zreach  \marking{\vec{\pi}}$.
    Recall that $\vec{\mu}'[\kappa], \vec{\pi}[\kappa] < \vec{\mu}[\kappa] = k$.
    By assumption, $\pn$ is strongly $\vec{\mu}'[\kappa]$-sound and strongly $\vec{\pi}[\kappa]$-sound.
    Therefore, $\marking{\vec{\mu}'} \reach \fmarked{\vec{\mu}'[\kappa]}$
    and $\marking{\vec{\pi}} \reach \fmarked{\vec{\pi}[\kappa]}$.
    Since the function $\marking{\cdot}$ is linear, we get
    \[
    \m = \marking{\vec{\mu}} = \marking{\vec{\mu}'} + \marking{\vec{\pi}}.
    \]
    This implies $\m \reach \fmarked{\vec{\mu}'[\kappa]} + \fmarked{\vec{\pi}[\kappa]} = \fmarked{k}$, which is a contradiction.
\end{proof}

\subsection{Reachability in $\Z$-bounded nets is in PSPACE}

Note that $\imarked{0} = \fmarked{0} = \vec{0}$. We will use these notations interchangeably depending on the emphasis.

\begin{lemma}\label{lem:zbounded}
  Let $\pn = (P, T, F)$ be a nonredundant workflow net and $k > 0$. If $\pn$
  is $\Z$-unbounded from $\imarked{k}$, then $\pn$ is not generalised
  sound.
\end{lemma}

\begin{proof}
  Since $\pn$ is $\Z$-unbounded from $\imarked{k}$, there exist
  $\m,\m'$ and $\pi$ such that $\m < \m'$ and
  $\imarked{k} \zreach \m \ztrans{\pi} \m'$. Thus,
  $\imarked{0} \ztrans{\pi} \m' - \m > \vec{0}$.
  For the sake of contradiction, assume that $\pn$ is generalised sound. It is strongly $k$-sound in particular for $k=0$ by \Cref{lem:unssound-unsound}, 
  so we have $\m'
  - \m \reach \fmarked{0}$, which contradicts the fact that
   $\post{t} \neq \vec{0}$ for all $t \in T$.
\end{proof}

\begin{lemma}\label{lem:z-unbounded}
  Let $\pn = (P, T, F)$ be a workflow net.
  Let $\m \in \N^P$ be a marking
  such that $\norm{\m} > \max(\norm{T},k)^2 \cdot (|P| + 2) \cdot |P|$.
  If $\imarked{k} \zreach \m$ then $\pn$ is $\Z$-unbounded.
\end{lemma}

\begin{proof}
  Let $\imarked{k} \ztrans{\sigma} \m$ for some $\sigma = t_1
  t_2 \cdots t_n$. We use the notation $\multiset{\cdot}$ for
  multisets, \eg\ $\multiset{a, a, b}$ contains two occurrences of $a$
  and one of $b$. Without loss of generality, assume
  that no submultiset of $\multiset{t_1, t_2, \ldots, t_n}$ sums to
  $\vec{0}$. Otherwise, we can shorten $\sigma$ by removing such a
  submultiset.

By \Cref{lem:extended-steinitz}, we can assume that $t_1, t_2, \ldots, t_n$ are ordered so that there exist $0 \le c_0 \le c_1 \le \ldots \le c_n$, where
\begin{gather*}
    \norm{\imarked{k} + \sum_{j=1}^i \effect{t_j} - c_i \m} \le \max(\norm{T},k) \cdot (|P| + 2),
\end{gather*}
for all $i \in [0..n]$.
Since $\norm{\m} > \max(\norm{T},k)^2 \cdot (|P| + 2) \cdot |P|$, we know that $n > \max(\norm{T},k) \cdot (|P| + 2) \cdot |P|$. By the pigeonhole principle, there must be $0 \le i_1 < i_2 \le n$ such that
\[
\imarked{k} + \sum_{j=1}^{i_1} \effect{t_j} - c_{i_1} \m = \imarked{k} + \sum_{j=1}^{i_2} \effect{t_j} - c_{i_2} \m.
\]
This is equivalent to
\[
\sum_{j=i_1 + 1}^{i_2} \effect{t_j} = (c_{i_2} - c_{i_1})\m.
\]
We have $(c_{i_2} - c_{i_1})\m \geq \vec{0}$ and, since no subset of $\multiset{t_1, t_2, \ldots, t_n}$ sums to $\vec{0}$, we have a strict inequality. Let $\vec{z} \defeq \sum_{j=i_1 + 1}^{i_2} \effect{t_j}$. We proved that $\imarked{0} \zreach \vec{z} > \vec{0}$, so $\pn$ is $\Z$-unbounded. 
\end{proof}

We are ready to prove the PSPACE membership of generalised soundness.

\begin{proof}[Proof of \cref{theorem:pspaceupper}]
Consider a workflow net $\pn = (P, T, F)$. By \Cref{claim:batch}, we can assume that $\pn$ is a nonredundant workflow net.
By \cref{lem:unssound-unsound} and \cref{thm:first-unsound-small}, to prove generalised soundness it suffices to prove that it is $k$-sound for all $k \le \norm{\pn}^{\poly(\abs{\pn})}$.

By \Cref{lem:zbounded} and \Cref{lem:z-unbounded}, if $\imarked{k} \reach \m$ and $\norm{\m} \ge C_k$ for some $C_k =  (\norm{\pn} + k)^{\poly(\abs{\pn})}$, then the net is unsound. Since we need to consider only $k \le \norm{\pn}^{\poly(\abs{\pn})}$, all constants $C_k$ are bounded exponentially and can be written in polynomial space.

Thus, to verify $k$-soundness we proceed as follows. First, we check if a configuration $\m$ such that $\norm{\m} \ge C_k$ can be reached. This can be easily performed in $\text{NPSPACE} = \text{PSPACE}$ as such a run would be witnessed by a sequence of configurations, such that each configuration can be stored in polynomial space. If such a configuration can be reached, then the algorithm outputs no. Otherwise, for every $\m \in \N^P$ such that $\norm{\m} < C_k$ one needs to verify whether $\imarked{k} \reach \m$ implies $\m \reach \fmarked{k}$. This can be done in $\text{coNPSPACE} = \text{coPSPACE} = \text{PSPACE}$.
\end{proof}

\subsection{PSPACE-hardness}

A \emph{conservative Petri net} is a 
Petri net $\pn = (P, T, F)$ such that transitions preserve the number of tokens. That is, for all $\m,\m' \in \N^P$, it is the case that
$\m \trans{} \m'$ implies $\sum_{p\in P} \m[p] =  \sum_{p \in P} \m'[p]$.
The \emph{reachability problem} for conservatrice Petri nets asks whether $\m \reach \m'$, given $\pn$, a source marking $\m$ and a target marking $\m'$.

\begin{restatable}{theorem}{thmPSpaceHard}\label{thm:pspace:hard}
  Generalised soundness is PSPACE-hard.
\end{restatable}

\begin{proof}
  We give a reduction from reachability in conservative Petri nets,
  which is known to be PSPACE-complete \cite{ConservativePN14}.

  Let $\pn = (P, T, F)$ be a conservative Petri net, and let $\m, \m'$ be the source and target markings.
  We define the constant $c \defeq \sum_{p\in P} \m[p] =  \sum_{p \in P} \m'[p]$.

  We construct a workflow net $\pn' = (P', T', F')$ such that $\pn'$ is generalised sound if and only if
  $\m \reach \m'$ in $\pn$. To do so, we extend $\pn$ with three new places
  $P' \defeq P \cup \set{\initial, \output, r}$. Places $\initial$ and $\output$ serve as 
decidated initial and final places, respectively.
  Place $r$ will be used to reset configurations.
  It could be merged with $\initial$, if not for the restriction that, in a workflow net, place $\initial$ cannot have any incoming arc.

  We define $T' \supseteq T$ by keeping the existing transitions and adding $3 + |P|$ new transitions.
  Namely:
  \begin{enumerate}
     \item transition $t_\initial$ defined by $\pre{t_\initial} \defeq \imarked{1}$, and $\post{t_\initial} \defeq \marked{r}{c}$,
     
     \item transition $t_{\m}$ defined by $\pre{t_{\m}} \defeq \marked{r}{c}$, and $\post{t_{\m}} \defeq \m$,
     
     \item transition $t_{\m'}$ defined by $\pre{t_{\m'} \defeq \m'}$, and $\post{t_{\m'}} \defeq \fmarked{1}$,
     
   \item transition $t_p$ defined by $\pre{t_p} \defeq \marked{p}{1}$, and $\post{t_p} \defeq \marked{r}{1}$.
  \end{enumerate}

  The first two transitions move a token from $\initial$ and create the marking $\m$. The third transition consumes $\m'$ and puts a token into $\output$. Transitions from the fourth group allow to move tokens from any place in the original Petri net $P$ to $r$. See \Cref{fig:pspace_hard} for a graphical presentation.

  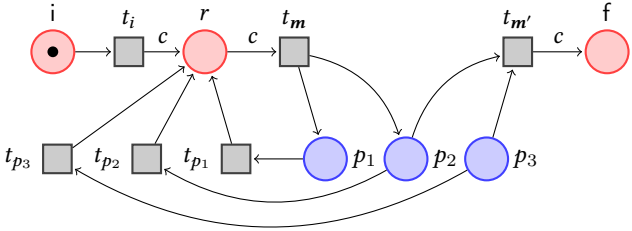
\begin{figure}[!h]
    \centering
\begin{tikzpicture}[node distance=1.25cm, transform shape, scale=0.95]
      \tikzstyle{place}=[circle,thick,draw=blue!75,fill=blue!20,minimum size=6mm]
      \tikzstyle{transition}=[rectangle,thick,draw=black!75,
        fill=black!20,minimum size=4mm]
        \tikzstyle{red place}=[place,draw=red!75,fill=red!20]
        
    \node [red place,tokens=1, label=above:{$\initial$}] (i) {};
    \node [red place, label=above:{$r$}] (q0) [right = 1.5cm of i]  {};
    \node [red place, label=above:{$\output$}] (q1) [right = 5cm of q0] {};
    
    \node [place, label={[xshift=-1.5pt]right:{$p_2$}}] (p2) at ($(q0)!0.5!(q1) + (0,-1.5)$) {};
    \node [place, label={[xshift=-1.5pt]right:{$p_1$}}, left=0.5cm of p2] (p) {};
    \node [place, label={[xshift=-1.5pt]right:{$p_3$}}, right=0.5cm of p2] (p3) {};
    
    \node [transition, label=above:{$t_i$}] (t1) at ($(i)!0.5!(q0)$) {}
    edge[pre] (i)
    edge[post] node[above] {$c$} (q0)
    ;
   \node [transition, label=above:{$t_{\m}$}] (tm) [right of=q0] {}
   edge[pre] node[above] {$c$} (q0)
   edge[post] (p)
   edge[post,bend left] (p2)
   ;

   \node [transition, label=above:{$t_{\m'}$}] (tm') [left of=q1] {}
   edge[pre, bend right] (p2)
   edge[pre] (p3)
   edge[post] node[above] {$c$} (q1)
   ;

   \node [left of=p, transition, label=left:{$t_{p_1}$}] (tp) {}
   edge[pre] (p)
   edge[post] (q0)
   ;
   
   \node [left of=tp, transition, label=left:{$t_{p_2}$}] (tp2) {}
   edge[pre,bend right] (p2)
   edge[post] (q0)
   ;
   
   \node [left of=tp2, transition, label=left:{$t_{p_3}$}] (tp3) {}
   edge[pre, bend right] (p3)
   edge[post] (q0)
   ;
   
\end{tikzpicture}
\caption{A workflow net $\pn'$ which is generalised sound iff
$m \reach m'$ in the conservative Petri net $\pn = (P, T, F)$. Here, $P = \set{p_1, p_2, p_3}$, $\m = \{p_1 \colon 1, p_2 \colon 1\}$, $\m' = \{p_2 \colon 1, p_3 \colon 1\}$ and $c = 2$. The original places are blue and the new places are red. We omit the original transitions (from $T$) in the picture.}\label{fig:pspace_hard}
\end{figure}

   It remains to show that $\pn'$ is correct. Suppose $\pn'$ is generalised sound. It must also be $1$-sound and in particular $\imarked{1} \reach \fmarked{1}$. Since $\pn$ is conservative, it is easy to see that $t_{\m}$ can be fired only if there are no tokens in $P$. Moreover, a token can be transferred to $\output$ only using $t_{\m'}$, which consumes $\m'$. Thus, we have $\m \reach \m'$ in $\pn$.

   The converse implication is shown in the appendix.
\end{proof}

\section{Structural soundness}
\label{sec:structural}
In this section, we establish the EXPSPACE-completeness of structural
soundness. Recall that the latter asks whether, given a workflow net,
$k$-soundness holds for some $k \geq 1$.

\subsection{EXPSPACE membership}

\begin{theorem}\label{theorem:structuralupper}
  Structural soundness is in EXPSPACE.
\end{theorem}

Let $\pn = (P, T, F)$ be a workflow net.
We define an $(\abs{T}+2\abs{P}+1) \times (\abs{T}+1)$-ILP, called ILP$_{\pn}^s$.
The variables are the same as for ILP$_{\pn}$ in \Cref{subsec:smallk}: $(\kappa, \tau_1, \dots, \tau_n)$, with the intuition that $\kappa$ denotes the number of initial tokens and $\tau_i$ the number of times the transitions are used. We will keep the notation $\ttau = (\tau_1, \dots, \tau_n)$ and the notation $\ttau[t]$ for $t \in T$. The inequalities are defined as follows:
\begin{enumerate}
\item $\imarked{\kappa} + \sum_{t \in T}\ttau[t] \cdot \effect{t} = \fmarked{\kappa}$ (expressed with $2\abs{P}$ inequalities);

\item $\ttau \geq \vec{0}$ ($\abs{T}$ inequalities);

\item and $\kappa > 0$.
\end{enumerate}
The first set of inequalities expresses that the effect of the transitions yields the final marking. The second type ensures that each transition is fired a nonnegative number of times. Finally the last one ensures that the initial marking has at least one token. The following is immediate.

\begin{claim}\label{claim:ilps}
There exists $k > 0$ such that $\imarked{k} \zreach \fmarked{k}$ if and only if there exists a solution $\vec{\mu}$ to ILP$_{\pn}^s$ such that $\vec{\mu}[\kappa] = k$.
\end{claim}

\begin{lemma}\label{lem:first-sound-small}
  Let $\pn = (P, T, F)$ be a nonredundant workflow net that is $k$-sound,
  and $i$-unsound for all $1 \leq i < k$. It is the case that $k \leq
  c + (\norm{T} + 2)^{\abs{T}} \cdot \max(\norm{T}, c) \cdot \abs{P}
  (\abs{P} + 2)$, where $c$ is the bound
  given by \Cref{lem:small-solution} for ILP$_{\pn}^s$.
\end{lemma}

\begin{proof}
  Towards a contradiction, suppose that $k > c + (\norm{T} + 2)^{\abs{T}} \cdot \max(\norm{T}, c)\cdot \abs{P} (\abs{P} +2)$.
  Consider ILP$_{\pn}^s$.
  Since $\pn$ is $k$-sound, there is a run $\imarked{k} \zreach \fmarked{k}$ and thus ILP$_s$ has a solution $\vec{\mu}$.
  By \Cref{lem:small-solution},
  we can assume that $\vec{\mu} \le \vec{c}$.
 
  By \Cref{claim:ilps}, $\imarked{\vec{\mu}[\kappa]} \zreach \fmarked{\vec{\mu}[\kappa]}$.
  By \Cref{lem:small-ell}, there exist $\ell \leq (\norm{T} + 2)^{\abs{T}} \cdot \max(\norm{T}, \vec{\mu}[\kappa]) \cdot \abs{P}(\abs{P} + 2)$ and $\m \in \N^P$
  such that $\imarked{\ell} \reach \m$ and
  $\imarked{\ell + \vec{\mu}[\kappa]} \reach \m + \fmarked{\vec{\mu}[\kappa]}$.
  Note that $\ell + \vec{\mu}[\kappa] < k$.
  Let $g \defeq k - (\ell + \vec{\mu}[\kappa]) > 0$.
  We have $\imarked{k} = \imarked{\ell + \vec{\mu}[\kappa] + g} \reach \imarked{g} + \m + \fmarked{\vec{\mu}[\kappa]}$. Since $\pn$ is $k$-sound, we have \[\imarked{g} + \m + \fmarked{\vec{\mu}[\kappa]} \reach \fmarked{\ell + \vec{\mu}[\kappa] + g}.\] Thus, $\imarked{g} + \m \reach \fmarked{\ell + g}$.
  We obtain
  \begin{multline*}
    \imarked{\ell + \vec{\mu}[\kappa] + g} \reach \imarked{\vec{\mu}[\kappa] + g} + \m \\
    {} \reach \imarked{\vec{\mu}[\kappa]} + \fmarked{\ell + g}.
  \end{multline*}
  Therefore, since $\pn$ is $k$-sound, it must be $\vec{\mu}[\kappa]$-sound (recall that tokens in $\output$ are never consumed).
   This contradicts the fact that $\pn$ is $i$-unsound for all $1 \leq i < k$.
\end{proof}

We may now prove \Cref{theorem:structuralupper}.

\begin{proof}[Proof of \Cref{theorem:structuralupper}]
  By \Cref{claim:batch}, we can assume that the input $\pn$ is a
  nonredundant workflow net. By \Cref{lem:first-sound-small}, it
  suffices to check if $\pn$ is $k$-sound for some value $k$ bounded
  exponentially in $\norm{\pn}$. First, we guess $k$, which can be
  written with polynomially many bits. Then, we test $k$-soundness in
  EXPSPACE via \Cref{cor:k-sound}.
\end{proof}

\subsection{EXPSPACE-hardness}

\begin{theorem}
  Structural soundness is EXPSPACE-hard.
\end{theorem}

\begin{proof}
  Let $\pn$ be a workflow net. We construct a workflow net $\pn'$ which is structurally sound iff $\pn$ is $1$-sound.
  We simply add a single new transition $t$ to $\pn$ with $\pre{t} \defeq \imarked{2}$ and $\post{t} \defeq \fmarked{1}$.
  We show that $\pn'$ is $k$-unsound for every $k \geq 2$.
  Towards a contradiction, suppose it is $k$-sound for some $k \geq 2$.

  Notice that $k$ cannot be even because $\imarked{k} \trans{t^{k/2}} \fmarked{k/2}$
  and $\output$ has no outgoing arcs, and hence $\fmarked{k/2} \not \reach \fmarked{k}$.
  Thus, it is the case that $k \ge 3$ is odd and $\imarked{k} \trans{t^*} \imarked{1} + \fmarked{\floor{k/2}}$.
  Since $\pn$ is $k$-sound, $\imarked{1} \reach \fmarked{\ceil{k/2}}$.
  But that implies $\imarked{k} \reach \fmarked{k \cdot \ceil{k/2}}$.
  Note that $k \cdot \ceil{k/2} > k$ as $k \ge 3$, which yields a contradiction since $\output$
  has no outgoing arcs to get rid of the extra tokens.

  To conclude, we observe that if the initial configuration in $\pn'$ is $\imarked{1}$, then
  it behaves like $\pn$ would, since $t$ will never be enabled, \ie\ it is not quasi-live.
  Thus, $\pn'$ is structurally sound if and only if $\pn$ is $1$-sound, and EXPSPACE-hardness follows from \Cref{thm:class-hard}.
\end{proof}


\section{Characterizing the set of sound numbers}
\label{sec:soundk}
Given a workflow net $\pn$, we define the set
$\sNums{\pn} \defeq \{k \geq 1 \mid \pn \text{ is $k$-sound}\}$. That
is, $\sNums{\pn}$ contains all the numbers for which $\pn$ is sound
(except $0$ which is trivial as any workflow net is $0$-sound). This
section is dedicated to providing and computing a representation of
$\sns$.

First, let us state a simple fact about $\sns$.\footnote{A similar
observation was made, but not explicitly stated, in~\cite[Lemma 2.2
and 2.3]{TM05}.}

\begin{lemma}
  The set $\sns$ is closed under subtraction with positive results.
\end{lemma}

\begin{proof}
  Let $g, k \in \sns$ be such that $g \gr k$.
 We show that $g - k \in \sns$.
  Since $k \in \sns$, we have
  $\imarked{g} = \imarked{k + (g - k)} \reach \fmarked{k} + \imarked{g -k}$.
  Since $N$ is $g$-sound, it must also be $(g-k)$-sound. So, $g-k \in \sns$.
\end{proof}

\begin{corollary}\label{cor:pk}
  There exist $p > 0$ and $k \in \N \cup \set{+\infty}$ such that 
  $\sns = \set{i\cdot p \mid 1 \le i < k}$.
\end{corollary}

By the above, $\sns$ is characterized by $p$ and $k$.
We thus say that a net is \emph{$(k, p)$-sound} if and only if
$\sns = \set{i\cdot p \mid 1 \le i < k}$.
Note that $k=0$ implies $\sns = \emptyset$.
Further, $k = +\infty$ if and only if $\sns$ is infinite.
Finally, a workflow net is generalised sound iff it is $(1,+\infty)$-sound; and it is
structurally sound iff there exist $p,k \geq 1$ such that it is $(k,p)$-sound.
We show that $k$ and $p$ can be computed.
This will rely on insights from the prior sections about the smallest 
numbers for which a net is unsound or sound.

\begin{theorem}\label{thm:sound-numbers}
Given a workflow net $\pn$, the numbers $p$ and $k$ that characterize $\sns$ are bounded by $\norm{\pn}^{\poly \bigO(\abs{\pn})}$, and hence can be represented with polynomially many bits. Given $\pn$, $p'$ and $k'$, the problem of deciding whether $\pn$ is $(k', p')$-sound is in EXPSPACE. Moreover, the algorithm computes $p$ and $k$ such that $\pn$ is $(k, p)$-sound.
\end{theorem}

\begin{proof}
Consider a workflow net $\pn$. By \cref{claim:batch}, we can assume that $\pn$ is nonredundant.
We will compute for which $p$ and $k$ the net $\pn$ is $(k,p)$-sound.
  By \cref{lem:first-sound-small},
  if $\sns \neq \emptyset$, then there exists
  $G \leq \norm{\pn}^{\poly \bigO(\abs{\pn})}$
  such that $\pn$ is $\ell$-sound for some $\ell \leq G$.
  By \cref{cor:k-sound}, it is possible to check $1$-soundness,
  $2$-soundness, \ldots, $G$-soundness in EXPSPACE.
  Thus, in EXPSPACE, we can identify the smallest
  $p$ such that $\pn$ is $p$-sound.

  It remains to compute $k$.
  Using \cref{lem:scale-net},
  we construct a net $\pn'$ which is
  $c$-sound if and only if $\pn$ is $cp$-sound for all $c > 0$.
  Thus, the smallest number $c$ for 
  which $\pn'$ is not $c$-sound is the smallest
  $c$ such that $\pn$ is not $cp$-sound.
  By \cref{thm:first-unsound-small},
  if $\sNums{\pn'} \neq \N \setminus \{0\}$ then there exists $G' \leq \norm{\pn}^{\poly \bigO(\abs{\pn})}$
  such that $\pn'$ is $c$-unsound for some $c \leq G'$.
  Thus, it suffices to check $1$-soundness, $2$-soundness, \ldots,
  $G'$-soundness to identify whether $k = +\infty$, or to compute the largest
  $k \in \N$ such that $\pn$ is $pk$-sound.
  By \cref{cor:k-sound}, $k$ can be computed in EXPSPACE.
\end{proof}

\section{Conclusion}
\label{sec:conlusion}
In this work, we settled, after around two decades, the complexity of
the main decision problems concerning workflow nets: $k$-soundness,
classical soundness, generalised soundness and structural
soundness. The first three are EXPSPACE-complete, while the latter is
PSPACE-complete and hence surprisingly simpler. We have further
characterised the set of sound numbers of workflow nets: they have a
specific shape that can be computed with exponential space.

As further work, we intend to study extensions of these problems in
the context of Petri nets. For example, a natural extension of
generalised soundness asks, given markings $\m$ and $\m'$, whether for
every $k \in \N$, every marking reachable from $k \cdot \m$ may lead
to $k \cdot \m'$. Contrary to workflow nets, a Petri net that
satisfies this property needs not to be bounded.

\balance
\bibliography{references}

\clearpage
\appendix
\section*{Appendix}
\label{sec:appendix}
\subsection*{Missing proofs of \Cref{sec:classical}}

\propCharacOneSound*

\begin{proof}
  Let $\pn = (P, T, F)$.

  $\Rightarrow$) For the sake of contradiction, suppose that
  $\pn_{sc}$ is unbounded. There exist markings $\m < \m'$ such that
  $\imarked{i} \trans{\pi} \m \trans{\pi'} \m'$ in $\pn_{sc}$. Let us
  assume, without loss of generality, that no marking repeats along
  the run. There are two cases to consider: either $\pi \pi'$ contains
  $t_{sc}$, or not.

  Let us argue that the first case cannot hold. For the sake of
  contradiction, assume it does. Let $\sigma t_{sc}$ be the shortest
  prefix of $\pi \pi'$ such that $\imarked{1} \trans{\sigma}
  \omarked{1} + \n \trans{t_{sc}} \imarked{1} + \n$ in $\pn$. If $\n =
  \vec{0}$, then we obtain a contradiction as no marking
  repeats. Otherwise, by $1$-soundness, we must have $\n \reach
  \vec{0}$, which is a contradiction as $\post{t} \neq \vec{0}$ for
  every $t \in T$.

  Thus, $\pi \pi'$ only contains transitions from $T$, which means we
  can reason about $\pn$ (rather than $\pn_{sc}$). By $1$-soundness,
  we have $\imarked{1} \reach \m \reach \omarked{1}$ in $\pn$. Since
  $\imarked{1} \reach \m'$ in $\pn$, altogether this yields
  \[
  \imarked{1} \reach \m'
  = \m + (\m' - \m) \reach \omarked{1} + (\m' - \m).
  \]
  By $1$-soundness, this means that $(\m' - \m) \reach \vec{0}$, which
  is impossible. Consequently, $\pn_{sc}$ is bounded from
  $\imarked{1}$.

  It remains to argue that $t_{sc}$ is live from $\imarked{1}$. Let
  $\imarked{1} \trans{\rho} \m$ in $\pn_{sc}$, where no marking
  repeats. We can assume that $t_{sc}$ does not appear in $\rho$ as it
  would mean that $\imarked{1}$ is repeated. Hence, $\imarked{1}
  \trans{\rho} \m$ in $\pn$. By $1$-soundness, we have $\m \reach
  \omarked{1}$, from which $t_{sc}$ is enabled as desired.

  $\Leftarrow$) Let $\imarked{1} \reach \m$ in $\pn$ (and so in
  $\pn_{sc}$). Since $t_{sc}$ is live from $\imarked{1}$, we have $\m
  \reach \omarked{1} + \n$ for some $\n \in \N^P$. If $\n > \vec{0}$,
  then we obtain $\imarked{1} \reach \omarked{1} + \n \trans{t_{sc}}
  \imarked{1} + \n$ which violates boundedness. Thus, $\n = \vec{0}$,
  and hence $\m \reach \omarked{1}$ as desired.
\end{proof}

\thmClassHard*

\begin{proof}
  Recall that in the proof within the main text, we have shown the
  implication from right to left of \Cref{eq:main}. It remains to
  prove the other direction.
    
  $\Rightarrow$) Suppose that $\m \reach \m'$. First, we prove that $\pn'$ is $1$-sound. Consider a configuration $\set{\initial \colon 1} \trans{\rho} \n$ for some run $\rho$. We consider cases depending on $\rho$ and $\n$.
    
    \emph{Case 0}: suppose that $t_{\mathrm{hard}}$ does not occur in $\rho$. If $\n = \set{\initial \colon 1}$ then $\n \trans{t_{\mathrm{simple}}t_{\mathrm{simple}2}} \set{\output \colon 1}$. Otherwise, either $\n = \set{\output \colon 1}$ or
    $$
    \n' = \set{p_{\mathrm{simple}} \colon 1} + \sum_{p \in P} \set{p \colon \norm{\pn}, \overline{p} \colon \norm{\pn}} \trans{\rho'} \n,
    $$
    where $\rho' \in T_1^*$. Since $T_1$ is reversible, we also have $\n \reach \n'$. This concludes this case as $\n' \trans{t_{\mathrm{simple}2}} \set{\output \colon 1}$. Notice that in the remaining cases the transitions $t_{\mathrm{simple}}$ and $t_{\mathrm{simple}2}$ can never be fired.
    
    \emph{Case 1}: suppose that $t_{\mathrm{hard}}$ occurs in $\rho$ but $t_{\mathrm{start}}$ does not. It is easy to see that $\rho \in t_{\mathrm{hard}}T_2^*$. Since $T_2$ is reversible and by \Cref{lem:count:gadget}~(1) , it is the case that
    \[
    \n \reach \set{s \colon 1, c \colon 1} \reach \set{f \colon 1, c \colon 1, b \colon c_n} + \sum_{\overline{p} \in \overline{P}} \set{\overline{p} \colon c_n}.
    \]
    Note that $t_{\mathrm{start}}$ is fireable from the latter configuration, and thus we can extend $\rho$ with that transition.
    We will discuss how to proceed from there in the next case.
    
    \emph{Case 2}: suppose that $t_{\mathrm{start}}$ occurs in $\rho$, but $t_{\mathrm{end}}$ does not.
    We have $\rho \in t_{\mathrm{hard}}T_2^*t_{\mathrm{start}}(T_1 \cup \set{t_{\m}, t_{\m'}, t_{\m'}^{-1}, t_{\mathrm{reach}}, t_{\mathrm{reach}}^{-1}})^*$. By \Cref{lem:count:gadget}~(2), the transition $t_{\mathrm{start}}$ in $\rho$ is fired between the configurations
\begin{multline*}
     \set{f \colon 1, c \colon 1, b \colon c_n} + \sum_{\overline{p} \in \overline{P}} \set{\overline{p} \colon c_n} \trans{t_{\mathrm{start}}} {} \\[-12pt] \set{p_{\mathrm{start}} \colon 1, b \colon c_n} + \sum_{\overline{p} \in \overline{P}} \set{\overline{p} \colon c_n}.
\end{multline*}
    If $\n$ is the latter configuration, then only $t_{\m}$ can be fired. Otherwise, the transition $t_{\m}$ is never available again. In any case, since the transitions $T_1 \cup \set{t_{\m'}, t_{\m'}^{-1}, t_{\mathrm{reach}}, t_{\mathrm{reach}}^{-1}}$ are reversible,
    $\n \reach \set{p_{\mathrm{inProgress}} \colon 1, b \colon c_n} + \m + \sum_{\overline{p} \in \overline{P}} \set{\overline{p} \colon c_n - \m[p]}$.
    
    Since $\m \reach \m'$, and by \Cref{lem:rev:en}, we know that
 \begin{multline*}
   \set{p_{\mathrm{inProgress}} \colon 1, b \colon c_n} + \m + \sum_{\overline{p} \in \overline{P}} \set{\overline{p} \colon c_n - \m[p]} \reach {} \\
   \set{p_{\mathrm{inProgress}} \colon 1, b \colon c_n} + \m' + \sum_{\overline{p} \in \overline{P}} \set{\overline{p} \colon c_n - \m'[p]} \trans{t_{\m'}} {} \\
   \set{p_{\mathrm{cover}} \colon 1, b \colon c_n} + \sum_{\overline{p} \in \overline{P}} \set{\overline{p} \colon c_n} \trans{t_{\mathrm{reach}}} \\
   \set{\heart{f} \colon 1, \heart{c} \colon 1, b \colon c_n} + \sum_{\overline{p} \in \overline{P}} \set{\overline{p} \colon c_n}.
 \end{multline*}
 By \Cref{lem:count:gadget}~(1), we get
 \begin{align*}
  \set{\heart{f} \colon 1, \heart{c} \colon 1, b \colon c_n} + \sum_{\overline{p} \in \overline{P}} \set{\overline{p} \colon c_n} \reach \set{s^\heartsuit \colon 1, c^\heartsuit \colon 1}.
 \end{align*}
 Then, by firing $t_{\mathrm{end}}$, we reach $\set{\output \colon 1}$ as required.
    
 \emph{Case 3}: suppose that $t_{\mathrm{end}}$ occurs in $\rho$. We divide the run into the fragments where $t_{\mathrm{start}}$ and $t_{\mathrm{end}}$ were used for the first time. It is the case that $\rho = \rho_1 t_{\mathrm{start}}t_{\m} \rho_2 t_{\mathrm{end}} \rho_3$, where 
\begin{multline*}
    \set{\initial \colon 1} \trans{\rho_1t_{\mathrm{start}}t_{\m}} \\ \set{p_{\mathrm{inProgress}} \colon 1, b \colon c_n} + \m + \sum_{\overline{p} \in \overline{P}} \set{\overline{p} \colon c_n - \m[p]} = \n_1,
\end{multline*}
    and $\rho_2$ consists only of transitions in 
\begin{align*}
    T_1  \cup T_3 \cup \set{t_{\m'}, t_{\m'}^{-1},t_{\mathrm{reach}}, t_{\mathrm{reach}}^{-1}}.
\end{align*}
    
    From \Cref{lem:count:gadget} (2) and definition of transitions in $T'$ for every reachable configuration $\set{\initial \colon 1} \reach \n'$ in $\pn'$
\begin{align}\label{eq:state}
 \n'[f] + \n'[p_{\mathrm{inProgress}}] + \n'[p_{\mathrm{cover}}] = 1.
 \end{align}
    In other words, there can be a token only in one of the three places.
    Recall that
    transitions in $T_1 \cup \set{t_{\m}}$ can be fired only if $\n'[p_{\mathrm{inProgress}}] = 1$ (as $\n'[p_{\mathrm{inProgress}}] = \n'[p_{\mathrm{canFire}}]$); and transitions in $\set{t_{\m'}^{-1},t_{\mathrm{reach}}}$ only if $\n'[p_{\mathrm{cover}}] = 1$.
    
    Consider the Petri net $\pn'' = (P'',T'',F'')$, which is as $(P',T', F')$ but with places reduced to $P'' \defeq P^\heartsuit$ (recall that $b = b^\heartsuit$) and transitions $T''$ the same as $T'$ projected onto $P''$. Notice that among transitions used in $\rho_2$ only transitions in $T_3 \cup \set{t_{\mathrm{reach}}, t_{\mathrm{reach}}^{-1}}$ have impact on $P''$. 
    Slightly abusing the notation, we keep the names of transitions in $\pn''$ (from $\pn'$); and similarly for configurations. Since $\n_1 \trans{\rho_2} \vec{v}$ such that $\vec{v} \ge \set{s^\heartsuit \colon 1, c^\heartsuit \colon 1}$ (as $t_{\mathrm{end}}$ can be fired afterwards), we get in $\pn''$:
    \[
    \set{b \colon c_n} \trans{\rho_2} \vec{v}.
    \]
    Thus by \cref{lem:count:gadget} (4) and \cref{eq:state} we get that
    \begin{align*}
     \rho_2 \in T_{1+}^* t_{\mathrm{reach}}T_3^*\left(t_{\mathrm{reach}}^{-1} T_{1+}^* t_{\mathrm{reach}}T_3^* \right)^*.
    \end{align*}
    where $T_{1+} =  T_1 \cup \set{t_{\m} \cup t_{\m}^{-1}}$.
    Consider the configurations in $\pn''$ after firing transitions in $\rho_2$ starting from $\n_1$. We claim that every time after $t_{\mathrm{reach}}$ was fired the configuration is $\vec{v}' = \set{b \colon c_n, f^\heartsuit \colon 1, c^\heartsuit \colon 1}$.
    Indeed, after the first time this is because $\set{b \colon c_n} \trans{t_{\mathrm{reach}}} \vec{v}'$ (recall that transitions in $T_{1+}$ have no impact on $\pn'')$.
    For the remaining cases, notice that between $t_{\mathrm{reach}}$ and $t_{\mathrm{reach}}^{-1}$ only transitions from $T_3$ are fired.
    By \Cref{lem:count:gadget}~(1 and 2) after firing $t_{\mathrm{reach}}^{-1}$ the configuration has to be $\set{b \colon c_n}$. Since transitions in $T_{1+}$ have no impact on $\pn''$ we are ready to conclude the proof.
    
    Let $\rho_{2} = \rho_{pre}t_{\mathrm{reach}}\rho_{suf}$ such that $\rho_{suf}$ does not contain $t_{\mathrm{reach}}$. We know that
    $$
    \n_1 \trans{\rho_{pre}} \set{b \colon c_n} \trans{t_{\mathrm{reach}}} \vec{v}' \trans{\rho_{suf}} \vec{v}
    $$
    in $\pn''$, and recall that $\vec{v} \ge \set{s^\heartsuit \colon 1, c^\heartsuit \colon 1}$. By \Cref{lem:count:gadget}~(3), we get $\vec{v} = \set{s^\heartsuit \colon 1, c^\heartsuit \colon 1}$ in $\pn''$.
    
    Let us analyse the vector $\vec{v}$ in $\pn'$. By \cref{eq:state} we know that $\vec{v}[p_{\mathrm{inProgress}}] = \vec{v}[p_{\mathrm{cover}}] = 0$.
    Note that
    \begin{align}
    \vec{v}[b] = \vec{v}[p] + \vec{v}[\overline{p}],
    \end{align}
    for all $p \in P$. This is easy to see since every transition preserves this equality for all reachable configuration.
    Thus $\vec{v}[p] = \vec{v}[\overline{p}] = 0$ for all $p \in P$. Notice that this concludes the proof as $\vec{v} \trans{t_{\mathrm{end}}} \set{\output \colon 1}$ and thus $\rho_3$ is an empty run and $\n = \set{\output \colon 1}$.
\end{proof}

\subsection*{Missing proofs of \Cref{sec:stuff}}

\lemSmallSolution*

\begin{proof}
Let $x_1, \ldots, x_n$ be the variables of $G$.
We define a $(3m \times (m+n))$-ILP $G'$ by slightly modifying $G$. For every inequality in the original ILP $G$, we add one fresh variable. We denote them $y_1,\ldots, y_m$. Now, recall that the inequalities in $G$ are of the form:
$\sum_{i = 1}^{n} \mat{A}[j,i] \cdot x_i \ge \vec{b}[j]$
for $j \in [1..m]$.
The ILP $G'$ is defined with the same inequalities, plus $m$ new equalities (recall that this requires $2m$ inequalities):
$\sum_{i = 1}^{n} \mat{A}[j,i] \cdot x_i - y_j = 0$ for $j \in [1..m]$.

Notice that, in solutions for $G'$, the variables $y_j$ are uniquely
determined by the valuation of $x_1, \ldots, x_n$. For convenience, we
will write $\vec{\mu}[x_i], \vec{\mu}'[y_j]$ when referring to the
components of solutions. For every $\vec{\mu} \in \nsolutions{G}$,
there is a unique $\vec{\mu}' \in \solutions{G'}$ such that
$\vec{\mu}'[x_i] = \vec{\mu}[x_i]$ for all
$i \in [1..n]$. Thus, since $\vec{b} \geq \vec{0}$, we have
$\nsolutions{G'} = \set{\vec{\mu}' \mid \vec{\mu} \in \solutions{G}}$.

We define $c$ as the constant from \Cref{lem:small-ilp} for $G'$. Now, let $\vec{\mu} \in \nsolutions{G}$ and let $\vec{\mu}' \in \nsolutions{G'}$ be its corresponding solution. By \Cref{lem:small-ilp}, there exists $\vec{\nu}' \in \nsolutions{G'}$ such that $\vec{\nu}' \le \vec{\mu}'$ and $\vec{\nu}' \le \vec{c}$. We define $\vec{\nu} \in \nsolutions{G}$ as the solution corresponding to $\vec{\nu}'$. It is clear that $\vec{\nu} \le \vec{\mu}$ and $\vec{\nu} \le \vec{c}$. For the remaining part, fix $j \in [1..m]$. Recall that $\vec{\nu}'[y_j] = \sum_{i = 1}^{n} \mat{A}[j,i] \cdot \vec{\nu}'[x_i]$ and $\vec{\mu}'[y_j] = \sum_{i = 1}^{n} \mat{A}[j,i] \cdot \vec{\mu}'[x_i]$. Thus, \[\sum_{i = 1}^{n} \mat{A}[j,i] \cdot \vec{\nu}[x_i] \le \sum_{i = 1}^{n} \mat{A}[j,i] \cdot \vec{\mu}[x_i],\] which concludes the proof.
\end{proof}

\lemExtendedSteinitz*

\begin{proof}
Let $c \defeq \norm{\vec{z}}$. We can assume that $b, c > 0$, as otherwise the proof follows immediately from \cref{lemma:steinitz}. We use the notation $\multiset{\cdot}$ for multisets, \eg\ $\multiset{a, a, b}$ contains two occurrences of $a$ and one occurrence of $b$. Let $V \defeq \multiset{\vec{x}_0,\ldots, \vec{x}_n}$ and let $W \defeq \multiset{-\vec{z}/c,\ldots,-\vec{z}/c}$ be $c$ copies of the same vector.
    Clearly, $\sum_{\vec{x} \in V \cup W} \vec{x} = \vec{0}$ and $\norm{\vec{x}} \le b$ for all $\vec{x} \in V'$.
    
    Consider the multiset of vectors $X \defeq \multiset{\vec{x}/b \mid \vec{x} \in V \cup W}$, \ie rescaled vectors from $V \cup W$. 
    By \cref{lemma:steinitz} we can order the vectors in $X$: $\vec{a}_0',\ldots, \vec{a}_{n+c}'$, so that
    \begin{gather*}
        \norm{\sum_{j=0}^i \vec{a}_{j}'} \leq d \text{ for all $i \in [0..n+c]$}.
    \end{gather*}
By scaling the vectors back, we get an order $\vec{a}_0,\ldots, \vec{a}_{n+c}$ on vectors in $V \cup W$ such that
    \begin{gather}\label{eq:order}
        \norm{\sum_{j=0}^i \vec{a}_{j}} \leq bd \text{ for all $i \in [0..n+c]$}.
    \end{gather}
Let $0 \le s_0 < s_1 < \ldots < s_n \le n+c$ be indices such that $V = \multiset{\vec{a}_{s_0},\ldots, \vec{a}_{s_n}}$. Fix some $i \in [0..n]$. Note that the number of remaining vectors is
\[
\abs{\multiset{\vec{a}_j \mid j \not \in \set{s_0,\ldots,s_n}, j \le i}} = s_i - i
\]
for all $i \in [0..n]$. By \eqref{eq:order},
\begin{gather}\label{eq:almost}
\norm{\sum_{j=0}^i \vec{a}_{s_j} - \frac{s_i-i}{c}\vec{z}} = \norm{\sum_{j=0}^{s_i} \vec{a}_j} \le bd.
\end{gather}
Thus, by defining $c_i \defeq \frac{s_i - i}{c}$, we get $0 \le c_0 \le c_1 \le \ldots \le c_n$. To conclude the lemma, it remains to show that we can reorder $\vec{a}_{s_0}, \ldots, \vec{a}_{s_n}$ so that $\vec{a}_{s_0} = \vec{x}_0$. Indeed, suppose that $\ell \in [0..n]$ is an index such that $\vec{a}_{s_\ell} = \vec{x}_0$. By~\eqref{eq:almost} and the triangle inequality
\begin{gather*}
\norm{\sum_{j=0}^i \vec{a}_{s_j} - \frac{s_i-i}{c}\vec{z} + \vec{a}_{s_\ell} - \vec{a}_{s_0}} \le b(d+2).
\end{gather*}
Therefore, by swapping $\vec{a}_{s_\ell}$ and $\vec{a}_{s_0}$ we obtain the desired permutation.
\end{proof}

\subsection*{Missing proofs of \Cref{sec:generalised}}

\propBatch*

\begin{proof}
  Redundancy can be checked with a saturation algorithm. We give a
  short proof that relies on the decidability of continuous Petri nets
  (which internally uses such a procedure).
  
  Let $\pn = (P, T, F)$. The continuous semantics of $\pn$ allows to
  scale transitions by nonnegative coefficients, and for markings to
  hold nonnegative rational values. More formally, in this context, a
  \emph{marking} is a vector $\m \in \Qpos^P$. Given $\lambda \in
  \Qpos$ and $t \in T$, we say that $\lambda t$ is \emph{enabled} in
  $\m$ if $\m - \lambda \pre{t} \geq \vec{0}$. If $\lambda t$ is
  enabled, then firing it leads to $\m' \defeq \m - \lambda \pre{t} +
  \lambda \post{t} = \m + \lambda \effect{t}$, which is denoted by $\m
  \ctrans{\lambda t} \m'$, or simply $\m \ctrans{} \m'$.

  We write $\m \creach \m'$ if $\m'$ can be reached in zero, one or
  several such steps from $\m$.

  Deciding continuous coverability (and, in fact, continuous
  reachability) can be done in polynomial time~\cite{FH15}. Thus, one
  can test the following in polynomial time, where $p \in P$:
  \begin{align}
    \exists \m' \in \Qpos^P : \m \creach \m' \land \m'[p] >
    0.\label{eq:cont:cov}
  \end{align}
  As $\m \creach \m'$ holds iff $k \m \reach k \m'$ holds for some $k
  > 0$, \eqref{eq:cont:cov} is equivalent to
  \begin{align}
    \exists k > 0, \m' \in \N^P : k \m \reach \m' \land \m'[p] >
    0.\label{eq:cont:cov2}
  \end{align}
  Hence, to test whether a place $p$ is nonredundant, it suffices to
  check \eqref{eq:cont:cov2} in polynomial time with $\m \defeq
  \imarked{1}$.
\end{proof}

\thmPSpaceHard*

\begin{proof}
  Recall the workflow net $\pn'$ constructed in the proof within the
  main text. We must show that $\pn'$ is generalised sound if and only
  if $\m \reach \m'$ in $\pn$. The implication from left to right has
  already been proven in the main text.
  
  For the converse implication, suppose that $\m \reach \m'$. Fix some $k$ and suppose $\imarked{k} \reach \vec{v}$. Notice that the transitions are defined in such a way that for every reachable configuration $\vec{v}$, the invariant $ck = \vec{v}[\initial] \cdot c + \sum_{p \in P \cup \set{r}} \vec{v}[p] + \vec{v}[\output] \cdot c$ holds. Thus, by repeatedly firing transitions $t_\initial$ and $t_p$, all tokens but those in $\output$ can be moved to $r$, i.e.\ \[\vec{v} \reach \marked{r}{(k-\vec{v}[\output]) \cdot c} + \fmarked{\vec{v}[\output]}.\] From there, to reach $\fmarked{k}$, it suffices to repeat $(k-\vec{v}[\output])$ times the following: fire $t_{\m}$; fire the run that witnesses $\m \reach \m'$; and fire $t_{\m'}$.
%
%
%
\end{proof}

\end{document}